\documentclass[letterpaper,11pt]{article}
\usepackage[utf8]{inputenc}

\pdfoutput=1

\usepackage{amsmath, amsthm, amssymb, thm-restate}
\usepackage{algorithmicx}
\usepackage[table,xcdraw]{xcolor}
\usepackage[ruled,vlined,linesnumbered]{algorithm2e}
\usepackage{setspace}
\usepackage{mathtools}
\usepackage[numbers]{natbib}
\usepackage{comment} 
\usepackage{tcolorbox} 
\usepackage{xfrac}
\usepackage{hyperref}
\usepackage{multirow}
\usepackage{caption}
\usepackage{bm}
\usepackage{newfloat}
\usepackage{enumitem}
\usepackage{dblfloatfix} 
\usepackage{wrapfig}
\usepackage{csquotes}

\usepackage{fancyhdr}

\usepackage[margin=1in]{geometry}

\allowdisplaybreaks

\definecolor{mygreen}{RGB}{10,150,110}
\definecolor{myred}{RGB}{150,10,20}

\hypersetup{
     colorlinks=true,
     citecolor= mygreen,
     linkcolor= myred
}

\renewcommand{\epsilon}{\varepsilon}

\DeclareMathOperator{\E}{\ensuremath{\normalfont \textbf{E}}}

\newcommand{\hiddencomment}[1]{}

\newcommand{\GMM}[0]{\ensuremath{\textup{GMM}}}

\newcommand{\ceil}[1]{{\left\lceil{#1}\right\rceil}}
\newcommand{\floor}[1]{{\left\lfloor{#1}\right\rfloor}}
\newcommand{\prob}[1]{\Pr\paren{#1}}
\newcommand{\card}[1]{\left\lvert#1\right\rvert}
\newcommand{\paren}[1]{\left( #1 \right)}
\newcommand{\bracket}[1]{\left[ #1 \right]}
\newcommand{\expect}[1]{\E\bracket{#1}}
\DeclareMathOperator*{\Exp}{\ensuremath{{\mathbb{E}}}}
\DeclareMathOperator*{\Prob}{\ensuremath{\textnormal{Pr}}}
\renewcommand{\Pr}{\Prob}

\DeclareMathOperator{\poly}{poly}

\usepackage[noabbrev,nameinlink]{cleveref}
\crefname{lemma}{Lemma}{Lemmas}
\crefname{theorem}{Theorem}{Theorems}
\crefname{property}{Property}{Properties}
\crefname{claim}{Claim}{Claims}
\crefname{result}{Result}{Results}
\crefname{definition}{Definition}{Definitions}
\crefname{observation}{Observation}{Observations}
\crefname{proposition}{Proposition}{Propositions}
\crefname{assumption}{Assumption}{Assumptions}
\crefname{line}{Line}{Lines}
\crefname{figure}{Figure}{Figures}
\creflabelformat{property}{(#1)#2#3}
\crefname{equation}{}{}
\crefname{section}{Section}{Sections}
\crefname{appendix}{Appendix}{Appendices}
\crefname{algCounter}{Algorithm}{Algorithms}
\Crefname{algCounter}{Algorithm}{Algorithms}

\newtheorem{theorem}{Theorem}

\newtheorem{lemma}{Lemma}[section]
\newtheorem{proposition}[lemma]{Proposition}

\newtheorem{definition}[lemma]{Definition}
\newtheorem{claim}[lemma]{Claim}

\newtheorem{remark}{Remark}
\newtheorem*{remark*}{Remark}

\definecolor{mylightgray}{RGB}{230,230,230}

\algnewcommand{\IIf}[2]{\textbf{if} #1 \textbf{then} #2}
\algnewcommand{\EndIIf}{\unskip\ \algorithmicend\ \algorithmicif}

\newenvironment{graytbox}{
\par\addvspace{0.1cm}
\begin{tcolorbox}[width=\textwidth,
                  boxsep=5pt,
                  left=1pt,
                  right=1pt,
                  top=2pt,
                  bottom=2pt,
                  boxrule=0pt,
                  arc=0pt,
                  colback=mylightgray,
                  colframe=black,
                  ]
}{
\end{tcolorbox}
}

\newenvironment{whitetbox}{
\par\addvspace{0.1cm}
\begin{tcolorbox}[width=\textwidth,
                  boxsep=5pt,
                  left=1pt,
                  right=1pt,
                  top=2pt,
                  bottom=2pt,
                  boxrule=1pt,
                  arc=0pt,
                  colframe=black,
                  colback=white
                  ]
}{
\end{tcolorbox}
}

\newcounter{algCounter}

\makeatletter
\renewcommand{\paragraph}{%
  \@startsection{paragraph}{4}%
  {\z@}{10pt}{-1em}%
  {\normalfont\normalsize\bfseries}%
}
\makeatother

\makeatletter
\patchcmd{\@algocf@start}
  {-1.5em}
  {0pt}
  {}{}
\makeatother

\title{Fully Dynamic Matching:\\ $(2-\sqrt{2})$-Approximation in Polylog Update Time}

\author{
Amir Azarmehr\\{\em Northeastern University} \and 
Soheil Behnezhad \\{\em Northeastern University} \and
Mohammad Roghani \\{\em Stanford University}
}

\date{}

\begin{document}

\maketitle

\thispagestyle{empty}
\begin{abstract}

    We study maximum matchings in fully dynamic graphs, which are graphs that undergo both edge insertions and deletions. Our focus is on algorithms that estimate the size of maximum matching after each update while spending a small time.

    \smallskip\smallskip
    An important question studied extensively is the best approximation achievable via algorithms that only spend $\poly(\log n)$ time per update, where $n$ is the number of vertices. The current best bound is a $(1/2+\epsilon_0)$-approximation for a small constant $\epsilon_0 > 0$, due to recent works of Behnezhad [SODA'23] ($\epsilon_0 \sim 0.001$) and Bhattacharya, Kiss, Saranurak, Wajc [SODA'23] ($\epsilon_0 \sim 0.006$) who broke the long-standing 1/2-approximation barrier. These works also showed that for any fixed $\epsilon > 0$, the approximation can be further improved to $(2-\sqrt{2}-\epsilon) \sim .585$ for bipartite graphs, leaving a huge gap between general and bipartite graphs.

    \smallskip\smallskip
    In this work, we close this gap. We show that for any fixed $\epsilon > 0$, a $(2-\sqrt{2}-\epsilon)$ approximation can be maintained in $\poly(\log n)$ time per update \emph{even in general graphs}. Our techniques also lead to the same approximation for general graphs in two passes of the semi-streaming setting, removing a similar gap in that setting.
\end{abstract}

{
\clearpage
\hypersetup{hidelinks}
\vspace{1cm}
\renewcommand{\baselinestretch}{0.1}
\setcounter{tocdepth}{2}
\thispagestyle{empty}
\clearpage
}

\setcounter{page}{1}
\section{Introduction}

We study the maximum matching problem in the fully dynamic setting.
Given a graph $G$ which undergoes both edge insertion and edge deletion updates,
the goal is to maintain a large matching while spending 
a small time per update. Denoting the number of vertices by $n$, the holy grail in dynamic graphs is to achieve algorithms with $\poly(\log n)$ update-time, as this would be polynomial in the size of each update (which can be represented with $\Theta(\log n)$ bits). Unfortunately, known conditional hardness results rule out any $O(n^{1 - \epsilon})$ update-time algorithm for maintaining an \emph{exact} maximum matching \cite{abboud2014popular, Dahlgaard16, henzinger2015unifying}. As such, much of the focus in the literature has been on \emph{approximate} maximum matchings \cite{arar2017dynamic, behnezhad2022new, behnezhad2020fully, bernstein2015fully, bernstein2016faster, bernstein2021framework, bhattacharya2021deterministic, bhattacharya2016new, bhattacharya2017fully, bhattacharya2018deterministic, charikar2018fully, grandoni2022maintaining, gupta2013fully, neiman2015simple, onak2010maintaining, Roghani2022beating,Behnezhad23,BKSW2023}.

For over a decade, we have had algorithms maintaining a greedy maximal matching, and thus a 1/2-approximation of maximum matching, in $\poly(\log n)$ time per update \cite{BaswanaGS-SJC18,Solomon-FOCS16,behnezhadDerakhshan19}. At the expense of using a larger polynomial in $n$ update time, it is known that the approximation can be improved using various {\em matching sparsifiers} developed in the literature \cite{GuptaPeng13,bernstein2016faster,behnezhad2022new}. However, the problem of maintaining a better-than-1/2-approximation in $\poly(\log n)$ time had remained open until last year where for a small  $\epsilon_0 > 0$, the concurrent works of \citet*{Behnezhad23} ($\epsilon_0 \sim 0.001$) and \citet*{BKSW2023} ($\epsilon_0 \sim 0.006$) achieved a $(1/2+\epsilon_0)$-approximation provided that the goal is to maintain just the size (and not the edge-set) of the matching. This state of affairs leaves two major open problems:
\begin{itemize}[itemsep=0pt,topsep=5pt]
    \item \emph{Can we also maintain the edges of a $1/2+\Omega(1)$ approximate matching in $\poly(\log n)$ time?}
    \item \emph{What is the best approximation of maximum matching size achievable in $\poly(\log n)$ time?}
\end{itemize}
Our focus in this work is on the latter question. 

\paragraph{A gap between bipartite and general graphs:} The algorithms of \cite{Behnezhad23,BKSW2023} have two phases. In the first phase, they maintain a maximal matching explicitly using the fast algorithms of \cite{BaswanaGS-SJC18,Solomon-FOCS16,behnezhadDerakhshan19}. In the second phase, they augment this maximal matching by employing the sublinear-time matching size estimator of \cite{Behnezhad21}. Although this leads to only slightly better than 1/2-approximation in general graphs, it is shown in \cite{Behnezhad23,BKSW2023} that it leads to a much better (almost) $(2-\sqrt{2} \sim 0.585)$-approximation if the input graph is bipartite. 

Such two-phase algorithms have also long been studied in the context of two-pass streaming algorithms \cite{KMM2012,KaleT17,EsfandiariHM16,Konrad18,KonradN-APPROX21} for which a similar gap between general and bipartite graphs has persisted. In particular, the state-of-the-art two-pass semi-streaming algorithm for bipartite graphs, by \citet*{Konrad18} from 2018, achieves the same (almost) $(2-\sqrt{2} \sim 0.585)$-approximation. However, despite attempts \cite{FeldmanS22,KaleT17} the best approximation for general graphs is 0.538 \cite{FeldmanS22}.

\paragraph{Our contribution:} In this paper, we close the aforementioned gap between bipartite and general graphs in both models. For dynamic graphs, we prove that the approximation can be improved to (almost) $2-\sqrt{2} \sim 0.585$, matching what was known for bipartite graphs and significantly improving the previous $0.506$ and $0.501$-approximations of \cite{BKSW2023,Behnezhad23} for general graphs.

\begin{graytbox}
\label{thm:dynamic-final}
\begin{restatable}[Formalized as \cref{thm:final-theorem-dynamic-detail}]{theorem}{mainthm}\label{thm:dynamic-final}
    For any fixed $\epsilon > 0$, there is an algorithm that maintains a $(2 - \sqrt{2} - \epsilon) \sim 0.585$-approximation of the size of the maximum matching in $\poly(\log n)$ worst-case update time even against adaptive adversaries.
\end{restatable}
\end{graytbox}

In the two-pass semi-streaming model, we show that the same (almost) $(2-\sqrt{2} \sim 0.585)$ approximation can be achieved for general graphs as well, matching what was known for bipartite graphs \cite{Konrad18} and significantly improving prior 0.531 and 0.538 approximations of \cite{KaleT17,FeldmanS22} for general graphs. We emphasize that our streaming algorithm does not just estimate the size of the maximum matching, but rather returns the edges of the matching as well. Additionally, unlike our dynamic algorithm, our streaming algorithm is deterministic.

\begin{graytbox}
\begin{restatable}{theorem}{twopassthm}\label{thm:two-pass-alg}
    For any fixed $\epsilon > 0$, there is a deterministic two-pass streaming algorithm that finds (the edges of) a $(2-\sqrt{2} - \epsilon) \sim 0.585$-approximate maximum matching using $O(n)$ space.
\end{restatable}
\end{graytbox}

While much of the technicality of our work is on the dynamic algorithm of \cref{thm:dynamic-final}, both the streaming algorithm of \cref{thm:two-pass-alg} and its analysis are simple and clean.

\paragraph{Going beyond $(2-\sqrt{2})$-approximations:} The $(2-\sqrt{2})$-approximation turns out to be a barrier in several settings, even for bipartite graphs. For instance, a work of \citet{HuangEtal-SODA19} establishes that no online matching algorithm under edge-arrivals (even allowing preemptions) can surpass $(2-\sqrt{2})$-approximations. While this is a different model than the ones considered in this work, it is in fact closely related to the  streaming setting. See, in particular, the paper of \citet[Section~1]{Kapralov-SODA21} who points out that his techniques \enquote{can probably be extended} to the construction of \citet{HuangEtal-SODA19}, ruling out \emph{single-pass} semi-streaming algorithms achieving better than $(2-\sqrt{2})$-approximations. We also refer interested readers to the paper of \citet*{KonradN-APPROX21} which is more specifically focused on \emph{two-pass} streaming algorithms and includes a discussion on beating $(2-\sqrt{2})$-approximations. In particular, they show this bound is tight for a certain class of algorithms and argue that \enquote{new techniques are needed in order to improve on the $(2-\sqrt{2})$ approximation factor}. Given this current landscape, we believe it is an important open question for future research to either go beyond $(2-\sqrt{2})$-approximations in the fully dynamic model or the two-pass streaming model, or alternatively, prove its impossibility.

\paragraph{Paper organization:} In \cref{sec:techniques}, we present an overview of our techniques both in the dynamic algorithm and in the streaming algorithm. As the streaming algorithm turns out to be much simpler, we first present the proof of \cref{thm:two-pass-alg} in \cref{sec:two-pass} as a warm-up to our techniques. We then incorporate the needed new ingredients and prove \cref{thm:dynamic-final} in \cref{sec:dynamic}.

\section{Technical Overview}\label{sec:techniques}

In this section, we give a brief overview of the technical challenges in designing our algorithms. We start with one of the existing algorithms for bipartite graphs and show why this algorithm does not perform well when the input graph is non-bipartite. We then discuss the new ingredients that we incorporate into our algorithm to achieve the same approximation ratio for non-bipartite graphs.

Let us consider the algorithm of \citet*{BKSW2023} for bipartite graphs which achieves a $0.585$-approximation. Let $b_v$ be the capacity of vertex $v$. We define a $b$-matching to be a collection of edges of $E$ such that each vertex $v$ has at most $b_v$ incident edges in the collection. Their algorithm has two main building blocks: (1) maintaining a maximal matching $M$, (2) a maximal $b$-matching $B$, in the bipartite graph between vertices of $V(M)$ and $V \setminus V(M)$. We let $\overline{V(M)} = V \setminus V(M)$ and let $G[V(M), \overline{V(M)}]$ be the induced bipartite graph between $V(M)$ and $\overline{V(M)}$. Since $B$ is between matched and unmatched vertices by $M$, it becomes challenging to dynamically maintain $B$. This difficulty arises due to the fact that updates in $M$ result in vertex updates in $G[V(M), \overline{V(M)}]$, and currently, there is no known algorithm capable of maintaining a matching under this type of update. To overcome this challenge, this algorithm uses the sublinear matching algorithm of \citet{Behnezhad21} to estimate the size of $B$. Then, it is possible to estimate the size of the maximum matching of $G$ as a function of only the sizes of $M$ and $B$. For a specific value of $b = 1 + \sqrt{2}$ and sufficiently large $k$, when the capacity of vertices in $V(M)$ is $k$ and the capacity of vertices in $\overline{V(M)}$ is $\ceil{kb}$ in the maximal $b$-matching, the algorithm achieves the approximation guarantee of $2 - \sqrt{2}$.

\paragraph{First challenge: parallel edges in the $b$-matching.} It is important to note that the algorithm of \cite{BKSW2023} allows the $b$-matching to include the same edge multiple times. Consider the graph depicted in \Cref{fig:figure1}. In this particular instance, our maximal matching $M$ contains only the edge $(u_1, u_2)$. Furthermore, the set $B$ contains $k$ duplicates of each of the two edges connected to $u_4$. Consequently, neither $M$ nor $B$ contains the dashed green edge, which implies that the approximation guarantee cannot be better than 0.5. In this example, it is vital for the algorithm to include unique edges in $B$ instead of selecting an edge multiple times. 

\begin{figure}[h]
    \centering
    \includegraphics[scale=0.7]{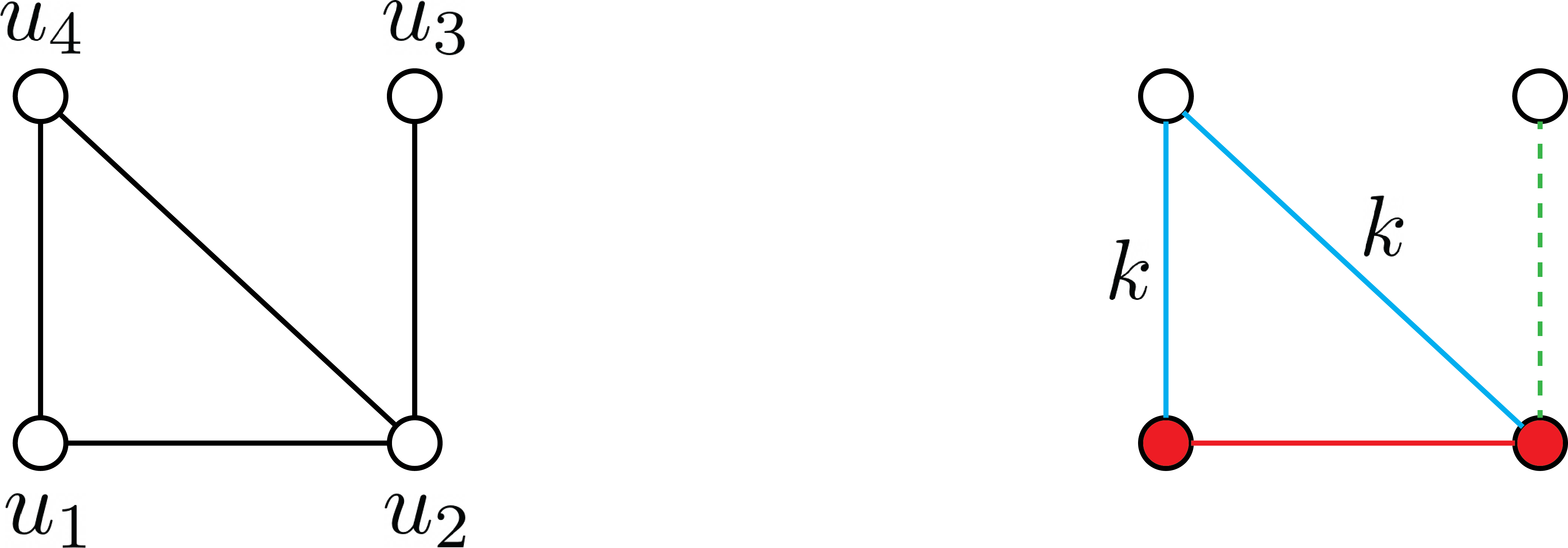}
    \caption{This figure shows a non-bipartite graph such that if the algorithm allows multi-edges in the $b$-matching, it fails to achieve a large approximation guarantee. In particular. The left figure is the input graph. The right figure is a possible output of the algorithm. Here, the red edge denotes the maximal matching $M$, the red vertices denote $V(M)$, the blue edges show the maximal $b$-matching $B$, and the dashed green edge is not in $M \cup B$. Moreover, the number of copies of each edge in $B$ is written next to it.}
    \label{fig:figure1}
\end{figure}

Since we store the $b$-matching explicitly in the streaming setting, it is easier to avoid including parallel edges. In \Cref{sec:two-pass}, we provide a novel analysis of this algorithm via fractional matchings which relies on blossom inequalities. We prove that if we restrict the $b$-matching to pick each edge at most once, then the algorithm indeed obtains a $(2-\sqrt{2})$-approximation even for general graphs. More specifically, we construct a fractional matching $x$ such that, except for the edges of two matchings, the fractional value of all edges is exceedingly small, i.e.\ $O(\epsilon^3)$. In the first matching, the fractional value of each edge is precisely $1-1/b$, while in the second matching, this value is at most $1/b$. This characterization of the fractional matching, as established in the analysis, helps to show that the blossom inequality holds for vertex sets of size at most $O(1/\epsilon)$. Consequently, we can utilize \Cref{prp:general-fractional-matching} to prove that $M \cup B$ contains an integral matching almost as large as the fractional matching.

However for the dynamic algorithm---specifically the part where the sublinear time algorithm of \cite{Behnezhad21} is used to estimate the size of the $b$-matching---it is important to \emph{allow} parallel edges. To see why this is true, note that the reduction from maximal $b$-matching to maximal matching relies on copying vertices with respect to their capacity. With this approach, one edge might be included in the maximal $b$-matching multiple times. To tackle this obstacle, we use the power of random greedy maximal matching when we find the maximal $b$-matching. Let $M^*_1$ be the maximum matching edges of $G$ that has exactly one matched endpoint in $M$. We show that for each edge $e \in M^*_1$, either it is included in $B$, or its capacity is saturated with many distinct incident edges. It is not hard to see that the following process is equivalent to constructing random greedy maximal matching: in each round, select one edge among the remaining edges uniformly at random, include it in the maximal matching, and remove both endpoints of the edge from the graph. Now, consider an edge $e$ in the matching $M^*_1$. Let's examine the rounds in which one of the incident edges to $e$ is chosen. If among these rounds, in many of them, the number of incident edges to $e$ is small, then with high probability, at least one copy of $e$ is going to be included in $B$. On the other hand, if in most of the rounds, the number of incident edges to $e$ is large, we anticipate the copies to be distributed evenly. This is enough for us to show that $M \cup B$ contains a large matching via constructing a large fractional matching and exploiting blossom inequalities in the analysis.

\paragraph{Second challenge: estimating the output size.} The second challenge arises when attempting to provide an estimate for the size of the maximum matching. Consider a graph that is a path of length three and its middle edge is in $M$ and another graph that is a triangle in which one of its edges is in $M$. Note that $|B|$ is equal in both graphs (including the duplicate edges), i.e.\ in both graphs $|B|$ is equal to $2k$ since the capacity of vertices in $V(M)$ is $2k$ and we can use all their capacities. However, the maximum matching size differs: the first graph has a maximum matching size of 2, while the second graph has a maximum matching size of 1. Consequently, if the algorithm relies solely on the size of $B$ to make its estimation, it cannot output a value greater than 1. This results in an approximation guarantee of 0.5 for the first graph.

Nevertheless, subgraph $G[M \cup B]$ contains a large matching, i.e.\ in this example it contains all maximum matching edges of the original graph. Thus, all we need is to accurately estimate $\mu(G[M \cup B])$, i.e.\ to $(1-\epsilon)$-approximate $\mu(G[M \cup B])$, to obtain the estimation for the maximum matching of the original graph. On the other hand, we cannot afford to construct the whole maximal $b$-matching $B$ explicitly since we are using sublinear algorithms to build $B$. Instead, we can access to incident edges of a vertex $v$ by spending $\widetilde{O}(n)$ time using the sublinear algorithm of \cite{Behnezhad21}. Additionally, since the maximum degree of subgraph $G[M \cup B]$ is a constant, the total number of vertices in a close neighborhood of vertex $v$ is a constant. This characteristic allows us to use maximum matching algorithms in the LOCAL model to estimate $\mu(G[M \cup B])$. By combining the previous two ideas, we can develop an algorithm that estimates $\mu(G[M \cup B])$ within a factor of $(1-\epsilon)$ in $\widetilde{O}(n)$ time, which we can afford.
\section{Preliminaries}\label{sec:preliminaries}

\paragraph{Notation:} Throughout the paper, we use $G = (V, E)$ to denote the input graph. We let $n$ denote the number of vertices in $G$. We use $\mu(G)$ to denote the size of the maximum matching of $G$. For $U \subseteq V$, we let $G[U]$ be the induced subgraph of $G$ on vertices $U$. Also, for $U \subseteq V$ and $U' \subseteq V$ such that $U \cap U' = \emptyset$, we use $G[U,U']$ to denote the induced bipartite subgraph between $U$ and $U'$. For $E'\subseteq E$, we let $G[E']$ be the subgraph of $G$ that is induced by edges $E'$.

\subsection{Computational Models}

\paragraph{Streaming Algorithms:} In the streaming setting, we assume that the edges of the input graph $G$  appear one by one and in an arbitrary order. The goal is to have a single pass over the stream and return a large matching of $G$ using a small space. Our focus is particularly on $O(n\cdot \poly \log n)$ space algorithms, which is also known as the semi-streaming setting. In the two-pass streaming setting, the algorithm is allowed to have two passes over the stream while the space complexity is the same.

\paragraph{Fully Dynamic Algorithms:} In the fully dynamic model, we have a graph on a fixed set of $n$ vertices. Edges then can be inserted or deleted from the graph. We study dynamic algorithms for estimating the size of matching. That is, the algorithm has to return an estimate $\widetilde\mu$ for the size of the maximum matching of $G$ after each update. For $0 \leq \alpha \leq 1$, we say an algorithm achieves an $\alpha$-approximation if $\alpha \cdot \mu(G)\leq \widetilde\mu \leq \mu(G)$ after each update. The objective is to achieve a high approximation ratio while spending a small time per update. 

We say the adversary---which issues the updates---is {\em oblivious} if he does not change the sequence of updates based on the algorithm's previous output. On the other hand, an {\em adaptive} adversary may choose the sequence of updates adaptively based on the algorithm's outputs.

\paragraph{Local Algorithms:} In the distributed LOCAL model \cite{Linial92}, each vertex of the graph hosts a processor and two processors can exchange unlimited messages in each round if their corresponding nodes are neighbors in the graph. The goal in LOCAL algorithms is to compute a property (e.g. a matching) of the underlying communication network in few rounds. The output is often returned in a distributed manner --- for example, each node outputs the neighbor to which it is matched, if any. In our algorithms, we only use the following by now well-known property of LOCAL algorithms: the existence of an $r$-round LOCAL algorithm for a problem implies that the output of each vertex is only a function of its $r$-hop neighborhood.

\subsection{Probabilistic Tools}
\begin{proposition}[Markov inequality] \label{markov}
If $X$ is a non-negative random variable,
then for any $a > 0$, it holds that
\[
\prob{X \geq a} \leq \frac{\E[X]}{a}.
\]
\end{proposition}

\begin{proposition}[Chernoff bound] \label{chernoff}
Let $X_1, \ldots, X_n$ be independent random variables
taking values in $[0, 1]$.
Let $X = \sum X_i$ and let $\mu = \Exp\bracket{X}$.
Then, for any $0 < \delta \leq 1$ and $0 < a \leq \mu$, we have 
\[
\prob{X \geq (1 + \delta)\mu} \leq \exp\paren{-\frac{\delta^2 \mu}{3}} \qquad \text{and} \qquad \prob{X \geq \mu + a} \leq \exp\paren{-\frac{a^2}{3 \mu}}.
\]
\end{proposition}

\subsection{Background on Matching Theory}

\begin{definition}
    Given a graph $G = (V, E)$,
    a $b$-matching with non-negative capacities $b : V \to \mathbb{Z}$, is a (multi)set $B$ of the edges in $E$,
    such that for any vertex $u \in V$, the number of edges in $B$ that are adjacent to $u$ is at most $b(u)$.
\end{definition}

\begin{definition}
    Given a graph $G = (V, E)$, a fractional matching of $G$
    is a set of weights $x : E \to [0,1]$ on the edges
    such that for any vertex $u$ the following condition is satisfied:
    \[
        \sum_{e \in \delta(u)} x_e \leq 1.
    \]
    For a vertex set $S \subseteq V$, we define 
    $x(S) = \sum_{e\in G[S]} x_e$.
\end{definition}

The following proposition is an application of the blossom inequality if we relax the constraints to only consider subsets of vertices with a size of at most $1/\epsilon$. We refer readers to \cite{azarmehr2023robust} for the proof and to section 25.2 of \cite{schrijver2003combinatorial} for a detailed discussion about blossom inequalities.

\begin{proposition}\label{prp:general-fractional-matching}
    Let $G$ be any graph,
    and let $x$ be a fractional matching on $G$
    such that for every vertex set $S \subseteq V$
    that $\card{S} < 1/\epsilon$,
    we have
    \[
        \sum_{e \in G[S]} x_e \leq \floor{\frac{\card{S}}{2}}.
    \]
    Then, it holds that $\mu(G) \geq (1 - \epsilon) \sum_e x_e$.
    We refer to the above inequality as the blossom inequality.
\end{proposition}

\paragraph{Random Greedy Maximal Matching:} Given an input graph $G = (V, E)$ and a permutation $\pi$ over the edges of $E$, a greedy maximal matching can be obtained by sequentially iterating over the edges in $E$ according to $\pi$ and adding each edge to the maximal matching if none of its adjacent edges have already been added. We let $\GMM(G, \pi)$ denote this maximal matching. If $\pi$ is chosen uniformly at random among all possible permutations over $E$, we call this maximal matching a random greedy maximal matching.

\section{Warm-up: The Two-Pass Streaming Algorithm}\label{sec:two-pass}

In this section, we introduce our two-pass streaming algorithm for general graphs. We adopt a well-established framework (see \cite{BKSW2023, KMM2012}) that has been commonly used in the literature for designing two-pass algorithms to find maximum matchings. The algorithm operates by first identifying a maximal matching in the initial pass, followed by obtaining a maximal $b$-matching in the subsequent pass. Specifically, during the first pass, we find a maximal matching denoted as $M$ within the graph $G$. We set $b = (1+\sqrt{2})$ and choose a large integer $k$ that we will specify later. Moving on to the second pass, we find a maximal $b$-matching $B$ in the bipartite graph $G[V(M), \overline{V(M)}]$ with a capacity of $k$ and $\ceil{kb}$ for the vertices in $V(M)$ and $\overline{V(M)}$, respectively. 
Crucially, we do not allow $B$ to contain multiple copies of an edge.
Finally, we output the maximum matching obtained from the union of $M$ and $B$.

Now let us discuss some of the key distinctions between our algorithm and the previous algorithms in this framework, which contribute to achieving an approximation ratio of $(2-\sqrt{2})$ for general graphs. The first notable difference lies in our approach to selecting the value of $k$. It is crucial to choose $k$ sufficiently large in our algorithm. This decision stems from the fact that if $k$ is too small, many edges in the maximal $b$-matching might not effectively contribute to augmenting the maximal matching $M$. To illustrate this, consider the scenario where $k = 1$ and $(u,v) \in M$. Now, for a vertex $w \in \overline{V(M)}$, it is possible that both edges $(u,w)$ and $(v,w)$ are included in our maximal $b$-matching. However, this inclusion of both edges does not lead to any length-three augmenting paths, which are necessary for expanding the matching. To overcome this issue, it becomes essential to let $k$ be sufficiently large. By doing so, we can avoid the problem mentioned above and ensure that the maximal $b$-matching includes edges that are truly beneficial for augmenting the maximal matching in the second pass.

In our proof, we construct a fractional matching and utilize \cref{prp:general-fractional-matching} to prove that $M \cup B$ has a large matching. It is worth noting that in a recent work by Bhattacharya, Kiss, Saranurak, and Wajc \cite{BKSW2023}, they also adopt the approach of selecting a sufficiently large value for $k$. However, a notable distinction arises in the second pass of their algorithm, specifically during the computation of the maximal $b$-matching. In their approach, they allow their algorithm to select an edge multiple times, whereas we do not. 
Allowing multiple copies of the same edge can lead to a situation where an edge belonging to the maximal matching has neighboring edges in the $b$-matching that are chosen multiple times, while certain other edges (i.e.\ edges of the optimal matching $M^*$) are not chosen at all. Consequently, this poses a challenge when attempting to construct a fractional matching that satisfies the blossom inequality for subsets of vertices with small sizes.

Our algorithm is formalized in \Cref{alg:two-pass}. In the rest of this section, we prove the approximation guarantee of \Cref{alg:two-pass}.

\paragraph{Notation:} Throughout this section, we let $G = (V, E)$ be the original graph. We use $M$ to show the maximal matching that our algorithm finds in the first pass of the stream. Let $V(M)$ be the endpoints of $M$ and $\overline{V(M)} = V \setminus V(M)$. Finally, let $M^*$ be an arbitrary maximum matching of $G$, $M^*_1 = M^* \cap (V(M) \times \overline{V(M)})$, and $M^*_2 = M^* \cap (V(M) \times V(M))$.

\begin{algorithm}
\caption{Two-pass Streaming Algorithm for General Graphs}
\label{alg:two-pass}

\textbf{Parameter:} let $b = 1 + \sqrt{2}$ and $k$ be an integer larger than $\frac{1}{b\epsilon^3}$.

\textbf{First Pass:} $M \gets$ maximal matching of $G$.  \algorithmiccomment{Finding maximal matching}

\textbf{Second Pass:} \algorithmiccomment{Finding $b$-matching}

Let $B = \emptyset$.
    
\For{$(u, v) \in G[V(M), \overline{V(M)}]$ where $u \in V(M)$}{
    \If{$\deg_B(u) < k$ and $\deg_B(v) < \ceil{kb}$}{
        $B \gets B \cup {(u, v)}$.
    }

}
    
\Return maximum matching of $M \cup B$.

\end{algorithm}

\twopassthm*

\paragraph{Proof outline:} To prove the theorem, we construct a fractional matching $x$ on $M \cup B$.
    In \cref{clm:two-pass-large-maximal}, we show the sum of $x$ on $M$ is at least $\paren{1 - \frac{1}{b}}\paren{\card{M_2^*} + \frac{1}{2}\card{M_1^*}}$.
    In \cref{clm:two-pass-large-bmatching}, we show that the sum of $x$ on $G[V(M), \overline{V(M)}]$ is (almost) at least $\frac{1}{b+1}\card{M_1^*}$.
    As a result, we can conclude $x$ has size (almost) at least $(2 - \sqrt{2})\mu(G)$.
    In \cref{clm:two-pass-fractional-gives-integral}, we show that $(1 - \epsilon)x$ satisfies the conditions of \cref{prp:general-fractional-matching} to prove $M \cup B$ has an integral matching (almost) as large as $x$.
    Finally, we put all this together to complete the proof.

First, we describe the construction of the fractional matching $x$ on $M \cup B$.
For all edges $e \in M$, we let $x_e = 1 - \frac{1}{b}$.
For all the edges $e \in B \setminus M_1^*$, we let $x_e = \frac{1}{\ceil{kb}}$.
Finally, for every edge $(u, v) \in B \cap M_1^*$ with $u \in V(M)$ and $v \in \overline{V(M)}$, 
we let $x_e = \frac{t}{\ceil{kb}}$ where $t$ is equal to $\min(k - \deg_B(u), \ceil{kb} - \deg_B(v))$.
Informally, for the analysis, we keep adding copies of $(u, v)$ to $B$
as long as $B$ remains a $b$-matching. This is a key fact in the proof of \cref{clm:two-pass-large-bmatching}. 

Notice that any vertex in $V(M)$ is adjacent to at most one edge from $M$
and $k$ edges from $B$.
Therefore, the sum of $x$ on the adjacent edges in $M$ is at most $1 - \frac{1}{b}$,
and the sum on adjacent edges in $B$ is at most $k \cdot \frac{1}{\ceil{kb}} \leq \frac{1}{b}$.
Similarly, any vertex in $\overline{V(M)}$ is adjacent to at most $\ceil{kb}$ edges from $B$, and is not adjacent to any edges of $M$.
Therefore, the sum of $x$ on the adjacent edges is  at most $\ceil{kb} \cdot \frac{1}{\ceil{kb}} = 1$.
Hence, $x$ is a fractional matching.

\begin{claim}\label{clm:two-pass-large-maximal}
    It holds that $x(M) \geq \paren{1 - \frac{1}{b}}\paren{\card{M_2^*} + \frac{1}{2}\card{M_1^*}}$.
\end{claim}
\begin{proof}
    Consider the vertices of $V(M)$ and how they are covered by the edges of $M^*$.
    There are three possibilities: 
    the vertex is also covered by $M^*_1$;
    the vertex is also covered by $M^*_2$; or
    the vertex is not covered by $M^*$.
    Notice that there are $\card{M_1^*}$ vertices of the first type,
    and $2\card{M^*_2}$ of the second type.
    Therefore, it holds:
    \[
    \card{V(M)} \geq \card{M_1^*} + 2\card{M_2^*},
    \]
    and since $\card{M} = \frac{1}{2}\card{V(M)}$:
    \[
    \card{M} \geq \card{M_2^*} + \frac{1}{2}\card{M_1^*}. 
    \]
    Given that the value of $x$ on the edges of $M$ is $1 - \frac{1}{b}$, the claim follows.
\end{proof}

\begin{claim} \label{clm:two-pass-large-bmatching}
    It holds that $x(B) \geq (1 - \epsilon) \frac{1}{b + 1}\card{M_1^*}$.
\end{claim}

\begin{proof}
    First, we introduce some definitions.
    For every edge $e \in B$ define $t_e$ equal to $x_e \cdot \ceil{kb}$,
    i.e.\ $t_e$ is an integer such that $x_e = \frac{t_e}{\ceil{kb}}$.
    Also, define $t(u)$ as the sum of $t$ on its adjacent edges, 
    that is:
    \[ t(u) = \sum_{e \in \delta_B(u)} t_e. \]
    Notice, for every $u \in V(M)$, it holds that $t(u) \leq k$,
    and for every $v \in \overline{V(M)}$, it holds that $t(v) \leq \ceil{kb}$.
    Furthermore, for every $(u,v) \in M_1^*$ with $u \in V(M)$ and $v \in \overline{V(M)}$,
    it holds that $t(u) = k$ or $t(v) = \ceil{kb}$.

    We use a charging argument.
    We order the edges of $B$ arbitrarily as $e_1, \ldots, e_N$,
    and let $B_i$ be the set of first $i$ edges.
    With respect to $B_i$, we define a potential $\phi_i$
    on every edge $(u,v) \in M_1^*$ with $u \in V(M)$ and $v \in \overline{V(M)}$:
    \[
    \phi_i(u, v) = \max\paren{\frac{\displaystyle\sum_{{e \in \delta_{B_i}(u)}}t_{e}}{k}, 
        \frac{\displaystyle\sum_{{e \in \delta_{B_i}(v)}}t_{e}}{\ceil{kb}}},
    \]
    which is equal to the maximum fraction of the used capacity on its endpoints.
    We also let:
    \[
    \phi_i = \sum_{(u,v) \in M_1^*} \phi_i(u,v).
    \]
    
    For an edge $e_i$, we charge it $c_i = \phi_i - \phi_{i-1}$.
    For each edge $e_i$, it holds that $c_i \leq t_{e_i} \cdot \paren{\frac{1}{k} + \frac{1}{\ceil{kb}}}$.
    Because it is adjacent to at most two edges of $M_1^*$,
    and it can increase the potential on either one by at most 
    $\frac{t_{e_i}}{k}$ and $\frac{t_{e_i}}{\ceil{kb}}$ respectively.
    Therefore, we have:
    \begin{equation} \label{eq:two-pass-eq1}
    \phi_N = \sum_{i = 1}^N c_i \leq \paren{\frac{1}{k} + \frac{1}{\ceil{kb}}} \sum_{i = 1}^N t_{e_i}.
    \end{equation}
    
    It also holds that $\phi_N(u,v) = 1$ for every $(u,v) \in M_1^*$.
    To show this, we examine two cases.
    If $(u,v) \in B$, that is $(u,v) = e_i$ for some $i$,
    then $\phi_j(u, v)$ is equal to one for all $j \geq i$.
    If $(u,v) \notin B$, then at the point in the stream that $(u, v)$ arrived,
    at least one endpoint must have been saturated by the edges of $B$,
    i.e.\ $\phi_i(u, v)$ is equal to one whenever $B_i$ includes 
    all the adjacent edges $(u, v)$ in $B$. Therefore, we have:
    \begin{equation} \label{eq:two-pass-eq2}
    \phi_N = \card{M_1^*}.
    \end{equation}

    Putting \eqref{eq:two-pass-eq1} and \eqref{eq:two-pass-eq2} together we get:
    \[
    \card{M_1^*} \leq \paren{\frac{1}{k} + \frac{1}{\ceil{kb}}} \sum_{i = 1}^N t_{e_i},
    \]
    or equivalently:
    \[
    \sum_{i = 1}^N t_{e_i} \geq \frac{k \cdot \ceil{kb}}{k + \ceil{kb}} \card{M_1^*}
    = \frac{k + kb}{k + \ceil{kb}} \cdot \frac{\ceil{kb}}{b+1}  \card{M_1^*}
    \geq (1 - \epsilon) \frac{\ceil{kb}}{b + 1} \card{M_1^*}
    \]
    Given the fact that $x(e_i) = \frac{t_{e_i}}{\ceil{kb}}$, it follows:
    \[
    x(B) = \frac{1}{\ceil{kb}} \sum_{i = 1}^N t_{e_i}
    \geq (1 - \epsilon) \frac{1}{b + 1}\card{M_1^*}. \qedhere
    \]
\end{proof}

\begin{claim} \label{clm:two-pass-fractional-gives-integral}
    $M \cup B$ contains an integral matching of size $(1 - \epsilon)^2\sum_e x_e$.
\end{claim}

\begin{proof}
    To prove the statement, we show that $(1-\epsilon)x$ satisfies the conditions of \cref{prp:general-fractional-matching}.
    That is, we prove $x$ satisfies $x(S) \leq \ceil{\frac{S}{2}}$ for every vertex set $S \subseteq V$ 
    of size at most $\frac{1}{\epsilon}$. 
    Notice that since $(1 - \epsilon)x$ is a fractional matching the inequality holds for any 
    set $S$ with an even size. 
    Also, notice that if $x$ satisfies the inequality for a set $S$, then so does $(1-\epsilon)x$.
    This leaves us with one case.

    Take a vertex set $S$ of size equal to $2s + 1 \leq \frac{1}{\epsilon}$ such that $x$ does not satisfy the condition,
    i.e.\ $s < x(S) \leq s + 1$.
    For any edge $e \in M$, we have $x_e \leq 1 - \frac{1}{b}$,
    for any edge $e \in B \cap M_1^*$, we have $x_e \leq \frac{1}{b}$, and
    for any edge $e \in B \setminus M_1^*$, we have $x_e \leq \frac{1}{\ceil{kb}}$.
    Given the fact that there are at most $s$ edges of $M$ and $B \cap M_1^*$ in $S$, we can conclude:
    \begin{align*}
        x(S) &= x(M) + x(B \cap M_1^*) + x(B \setminus M_1^*) \\
        &\leq \paren{1 - \frac{1}{b}}\card{M} + \frac{1}{b} \card{B \cap M_1^*} + \frac{1}{\ceil{kb}}\card{S}^2  \\
        &\leq \paren{1 - \frac{1}{b}}s + \frac{1}{b} s + \frac{1}{\ceil{kb}}\frac{1}{\epsilon^2} \\
        &\leq s + \epsilon \tag{$k \geq \frac{1}{b\epsilon^3}$}
    \end{align*}
    Therefore, we have:
    \[
    (1-\epsilon)x(S) \leq (1-\epsilon)(s+\epsilon) \leq s + \epsilon - s\epsilon - \epsilon^2 \leq s - \epsilon^2.
    \]
    The claim follows from applying \cref{prp:general-fractional-matching} to $(1 - \epsilon)x$.
\end{proof}

\begin{proof}[Proof of \cref{thm:two-pass-alg}]
    First, we use \cref{clm:two-pass-large-maximal,clm:two-pass-large-bmatching} to show $\sum_e x_e \geq (1 - \epsilon)(2 - \sqrt{2})\mu(G)$.
    It holds that:
    \begin{align*}
        \sum_e x_e &= x(M) + x(B) \\
        &\geq \paren{1 - \frac{1}{b}}\paren{\card{M_2^*} + \frac{1}{2}\card{M_1^*}} + (1 - \epsilon) \frac{1}{b + 1}\card{M_1^*} \tag{\cref{clm:two-pass-large-maximal,clm:two-pass-large-bmatching}} \\
        &\geq 
        (1-\epsilon)\bracket{\paren{1 - \frac{1}{b}}\card{M_2^*} 
        + \paren{\frac{1}{2} -\frac{1}{2b} + \frac{1}{b+1}} \card{M_1^*}}.
    \end{align*}
    Since $b = 1 + \sqrt{2}$,
    we have $1 - \frac{1}{b} = \frac{1}{2} - \frac{1}{2b} + \frac{1}{b + 1} = 2 - \sqrt{2}$.
    Therefore,
    \[
    \sum_e x_e \geq (1 - \epsilon)(2 - \sqrt{2})(\card{M^*_1} + \card{M^*_2})
    = (1 - \epsilon)(2 - \sqrt{2}) \mu(G).
    \]
    
    To complete the proof, we note that by \cref{clm:two-pass-fractional-gives-integral},
    $M \cup B$ contains a matching of size $(1-\epsilon)^3(2 - \sqrt{2})\mu(G) \geq
    (1 - \epsilon) ^3\cdot .585 \cdot \mu(G)$.
    Also, \Cref{alg:two-pass} stores $O(n)$ edges for $M$ and $O(n \poly \frac{1}{\epsilon})$ edges
    for $B$. Hence, it uses space $O(n \poly \frac{1}{\epsilon})$.
    Replacing $\epsilon$ by $\frac{\epsilon}{3}$ gives the theorem.
\end{proof}
\newcommand{\tH}{{\widetilde{H}}}
\newcommand{\tB}{{\widetilde{B}}}

\section{The Fully Dynamic Algorithm} \label{sec:dynamic}

In this section, we show how we can turn our two-pass streaming algorithm in \Cref{sec:two-pass} into a fully dynamic algorithm with polylogarithmic update time. More formally, we prove the following theorem.

\begin{theorem}\label{thm:final-theorem-dynamic-detail}
For any $\epsilon > 0$, there is a fully dynamic algorithm that maintains a $(2 - \sqrt{2} - \epsilon) \sim 0.585$-approximation of the size of maximum matching in $2^{\poly(1/\epsilon)}\cdot \poly(\log n)$ worst-case update time in general graphs. The algorithm is randomized but works against adaptive adversaries.
\end{theorem}

Prior to delving into the algorithm and proofs, we define a setting known as a {\em semi-dynamic} setting. In this context, an algorithm $\mathcal{A}$ is categorized as semi-dynamic if it only generates an estimation of the maximum matching size when prompted with a query. By known reductions \cite{Behnezhad23,BKSW2023,Kiss22}, such semi-dynamic algorithms can be transferred into fully dynamic algorithms that are capable of maintaining the maximum matching size continuously and not only upon receiving a query. The following lemma from \cite{Behnezhad23}, formalizes this:

\begin{proposition}[Lemma 4.1 in \cite{Behnezhad23}]\label{prop:semi-dynamic-to-fully}
For a fully dynamic graph $G$ and $\epsilon > 0$, suppose there is a data structure $\mathcal{A}$ that takes $U(n)$ worst-case time per update to $G$ and provides an estimate $\widetilde\mu$ in $Q(n, \epsilon)$ time upon being queried, satisfying $\alpha \cdot \mu(G) - \epsilon n \leq \E[\widetilde\mu] \leq \mu(G)$.  Then, there exists a randomized data structure $\mathcal{B}$ that maintains an estimate $\widetilde\mu'$ such that, throughout the updates, $(\alpha - \epsilon) \cdot \mu(G) \leq \widetilde\mu' \leq \mu(G)$ with high probability. Additionally, $\mathcal{B}$ has a worst-case update time of $O\left(\left(U(n) + \frac{Q(n, \epsilon^2)}{n}\right) \cdot \poly(\log n, 1/\epsilon)\right)$. Moreover, if $\mathcal{A}$ can handle an adaptive adversary, then $\mathcal{B}$ can as well.
\end{proposition}

The proof of the aforementioned lemma is built upon the concept of \enquote{vertex sparsification,} which has been previously used in the literature \cite{assadiKhannaYang2019, Kiss22}. In this work, we show that there is a semi-dynamic algorithm that for any general graph $G$ achieves $U(n) = \text{poly}(\log n)$, $Q(n, \epsilon) = n\cdot 2^{\poly(1/\epsilon)}\cdot \poly(\log n)$, $\alpha = 2- \sqrt{2}$, and works against adaptive adversaries. Plugged into \cref{prop:semi-dynamic-to-fully}, this implies our desired \cref{thm:dynamic-final}.

The first step to design such a dynamic algorithm is to simulate the first pass of the original streaming algorithm. The problem of maintaining maximal matching in a dynamic setting has been extensively studied in the field and there exists a rich literature that leads to fully dynamic algorithms with polylogarithmic update time (see \cite{behnezhadDerakhshan19, bernsteinForster21, Solomon-FOCS16}). We use the following result as one of the building blocks of our algorithm.

\begin{proposition}[\cite{behnezhadDerakhshan19}]\label{prop:dynamic-maximal}
There exists a data structure that maintains a maximal matching in a fully dynamic graph with $\poly(\log n)$ worst-case update time against an oblivious adversary.
\end{proposition}

In order to simplify our algorithm, we incorporate the assumption of having an oblivious adversary as stated in \Cref{prop:dynamic-maximal}. However, it is important to note that this assumption is only employed at this specific point in the outline of our algorithm. Toward the end of \cref{sec:dynamic-approx}, we will explain how we can eliminate this assumption.

Let $M$ be the maximal matching that we maintain and $H$ be the induced bipartite subgraph between matched and unmatched vertices, i.e.\ $G[V(M), \overline{V(M)}]$. To simulate the second pass of our streaming algorithm, we incorporate a sublinear algorithm for estimating the size of the maximum matching. A similar approach was utilized in \cite{Behnezhad23, BKSW2023} to achieve an equivalent approximation ratio for bipartite graphs. For a more comprehensive understanding of the reduction from a dynamic matching algorithm to a sublinear matching algorithm, refer to \cite{Behnezhad23, BKSW2023}. In the second pass of the streaming algorithm for bipartite graphs, even if we select the same edge multiple times in the maximal $b$-matching, we can still demonstrate the same approximation guarantee. However, when dealing with general graphs, it is necessary to choose distinct edges in order to attain a significant approximation ratio. Let $B$ be a maximal $b$-matching in $H$. In \Cref{sec:two-pass}, we proved that $M \cup B$ contains a 0.585-approximate matching. A second technical challenge arises at this point. Unlike the algorithm presented in \cite{BKSW2023} for bipartite graphs, where the size of $B$ is adequate for estimation, we now need to estimate the size of the maximum matching of $M \cup B$ accurately. More formally, we need a $(1-\epsilon)$-approximation of $\mu(G[M \cup B])$ to achieve our approximation guarantee.

Our primary technical contribution in implementing our two-pass streaming algorithm in the fully dynamic setting involves addressing the above two challenges. If we are not restricted in selecting distinct edges for maximal $b$-matching, we can simplify the process by creating multiple copies of $\overline{V(M)}$ (specifically, $\ceil{kb}$ copies) and $V(M)$ ($k$ copies). This reduction allows us to convert the maximal $b$-matching into an instance of maximal matching, for which we already have a fast sublinear algorithm available \cite{Behnezhad21} (similar to the approach used in \cite{BKSW2023}).

\begin{proposition}[\cite{Behnezhad21, Behnezhad23}]\label{prop:sublinear-maximal}
    Let $\epsilon > 0$, $v$ be a random vertex in graph $G$, and $\pi$ be a random permutation over edges of $G$. There exists an algorithm that determines if $v$ is matched in $\GMM(G, \pi)$ that works in $\widetilde{O}(n/\epsilon)$ expected time with a success probability of $1 - \epsilon$. Moreover, if $v$ is matched, the algorithm returns the matching edge.
\end{proposition}

However, with this reduction, the possibility of selecting an edge multiple times arises, preventing us from obtaining the desired approximation guarantee for general graphs. It is worth noting that the algorithm in \cite{Behnezhad21} estimates the size of the randomized greedy maximal matching. To overcome the aforementioned challenge, we leverage the observation that when we run the randomized greedy maximal matching on the maximal $b$-matching instance, each edge in $M^*_1$  will either be selected at least once or has one endpoint that is nearly saturated with distinct edges in the maximal $b$-matching. This observation allows us to achieve the same approximation guarantee, disregarding some dependence on $\epsilon$.

For the second challenge, we need to design an oracle that, given a vertex $v$ as input, can determine whether $v$ is part of an approximately optimal maximum matching of $M \cup B$. We will then apply this oracle to several randomly selected vertices to estimate $\mu(G[M \cup B])$. To design the oracle, we can exploit the fact that the maximum degree of $G[M \cup B]$ is constant. This enables us to utilize existing LOCAL algorithms for maximum matching, as the number of vertices within a bounded distance from the queried vertex is at most a certain constant.

\begin{proposition}[\cite{harris2020}]\label{prop:local-algorithm}
    For $\epsilon > 0$, there exists a $O(\epsilon^{-3} \log \Delta)$-round LOCAL algorithm that outputs $(1-\epsilon)$-approximate maximum matching in expectation.
\end{proposition}

Note that we do not have direct access to the adjacency matrix of graph $G[M \cup B]$. However, we can utilize \Cref{prop:sublinear-maximal} to identify all the maximal $b$-matching edges of a vertex $v$. This allows us to obtain the neighbors of $v$ in $M \cup B$ within a time complexity of $\widetilde{O}_\epsilon(n)$. Consequently, by spending $\widetilde{O}_\epsilon(n)$ time, we can obtain all the vertices that are at a distance of $O(\epsilon^{-3} \log \Delta)$ from a given vertex.

In the rest of this section, we provide formal proof of the approximation guarantee and running time of \Cref{alg:dynamic}.

\paragraph{Notation:}
Throughout this section, we let $G = (V, E)$ be the original graph that undergoes edge deletion and insertion. We use $M$ to show the maximal matching that our algorithm maintains. Let $V(M)$ be the endpoints of $M$ and $\overline{V(M)} = V \setminus V(M)$. Let $H \coloneqq G[V(M), \overline{V(M)}]$ and $\tH$ be the graph constructed by having $k$ copies of vertices of $V(M)$ and $\ceil{kb}$ copies of $\overline{M}$. Additionally, let $\tB$ denote the random greedy maximal matching of $\tH$ and $B$ be the corresponding $b$-matching on $H$. Finally, let $M^*$ be an arbitrary maximum matching of $G$, $M^*_1 = M^* \cap (V(M) \times \overline{V(M)})$, and $M^*_2 = M^* \cap (V(M) \times V(M))$.

\begin{algorithm}\caption{Semi-Dynamic Algorithm for General Graphs}\label{alg:dynamic}
    
    Let $M$ be the maximal matching of $G$ that we maintain using \Cref{prop:dynamic-maximal}. 

    Let $H =G[V(M), \overline{V(M)}]$, and $\tH$ be the auxiliary graph based on $H$ ($\tH$ is not constructed explicitly).

    Let $\pi$ be a random permutation over edges of $\tH$ ($\pi$ is not constructed explicitly).
    
    Let $\tB$ be a random greedy maximal matching on $\tH$ with respect to $\pi$, and let $B$ be the corresponding $b$-matching on $H$ ($\tB$ and $B$ are not constructed explicitly).

    Sample $r = 24\epsilon^{-2} \log n$ random vertices $v_1, \ldots, v_r$ from $V$. \label{ln:number-of-samples}

Let $X_i$ be the indicator variable for the event that $v_i$ is matched in the $(1-\epsilon)$-approximate maximum matching of $M \cup B$ computed via \Cref{prop:local-algorithm}.

Let $X = \sum_{i=1}^r X_i$ and $\widetilde \mu = \frac{nX}{2r} - \frac{\epsilon n}{2}$. \label{ln:mu-estimate}
\end{algorithm}

\subsection{Approximation Ratio} \label{sec:dynamic-approx}

This subsection is devoted to proving the following claim, that is, $G[M \cup B]$ approximates the maximum matching of $G$. 

\begin{claim} \label{clm:dynamic-approx}
    It holds that 
    $(2 - \sqrt{2} - \epsilon)\mu(G) \leq \expect{\mu(G[M \cup B])} \leq \mu(G)$.
\end{claim}

\begin{remark}
Throughout this subsection, to make the proof more simple, we assume that $M$ is a maximal matching in $G$, and $B$ is a maximal $b$-matching in $H$.
Whereas, $M$ is an almost maximal matching and the edges of $B$ may be \enquote{missed} with probability $\epsilon$.
Toward the end, we show how these assumptions can be lifted.
\end{remark}

To prove the claim, we adopt a similar strategy to \cref{sec:two-pass}.
We construct a large fractional matching $x$ on $M \cup B$
and then show $M \cup B$ has an integral matching almost as large as $x$.
Finally, we conclude that $\mu(G[M \cup B])$ approximates $\mu(G)$.
Since we now allow $B$ to contain multiple copies of each edge and $G$ may be non-bipartite, the same argument no longer works.
However, we can show that $B$ has certain properties and the claim still holds.
We prove:
\begin{claim} \label{clm:dynamic-distributed-bmatching}
    For every edge $e \in M^*_1$, the following holds
    with probability at least $1 - \epsilon$:
    Either $e$ appears in $B$,
    or for an endpoint $u$ of $e$,
    there exists a multiset of edges $F \subseteq B$
    such that $\card{F} \geq (1-2\epsilon) b(u)$ and
    no edge appears in $F$ more than $\epsilon^3\ceil{kb}$ times. 
    Where $b(u)$ is the capacity of $u$ in the $b$-matching,
    i.e.\ $k$ if $u \in V(M)$ and $\ceil{kb}$ if $u \in \overline{V(M)}$.
\end{claim}

We can view the process of finding a random greedy maximal matching in $\tH$
as follows:
In each step,
a corresponding copy of edge $(u, v) \in H$ is picked
with probability proportional to its weight $w(u,v) = r(u) r(v)$.
Where $r(u)$ is the number of remaining vertices in $\tH$ corresponding to $u$.
Afterward, both endpoints are deleted from $\tH$.
Note that $r(u)$ can also be regarded as the unused capacity of $u$.

Now, we fix an edge $(u, v)$ in $M^*_1$
and define the edge set $I \subseteq E(H)$ equal to $(u, v)$ plus the set of edges adjacent to $(u, v)$, and let $I' = I \setminus \{(u, v)\}$.
The process alternates between picking an edge from $I$ and picking a number of edges from $E(H) \setminus I$.
We focus on the steps where an edge from $I$ is picked
and model the steps that happen outside of $I$ with an adversary.
Note that this adversary does not really exist.
It only represents the complex process of picking edges from $E(H) \setminus I$.

More formally, we model the process with a chain of steps, each composed of two parts.
The first part corresponds to picking edges from $E(H) \setminus I$,
and the second part corresponds to picking an edge from $I$.
In the first part, the adversary is given the number of times each edge of $I$ has been picked so far, hereafter referred to as the load of the edges.
It will then decide the weights $w$ on $I'$ for the next part
(the weight of $(u, v)$ is uniquely determined by the load of the edges in $I$).
In the second part, one of the edges in $I$ is picked at random with probability proportional to its weight.

Note that in the original process, the weights $w$ should satisfy certain constraints. For example, the weights should be non-increasing throughout the process and the weight of any edge $(u', v')$ should be determined by the underlying values $r(u')$ and $r(v')$ which in turn have constraints of their own  based on the loads.
We allow for a stronger adversary by disregarding many of these constraints
and imposing only a few of them. For now, we impose:
\begin{enumerate}
    \item In every step, every edge has an integer weight in $[0, k \cdot \ceil{kb}]$; and
    \item the total weight of $I'$,
    hereafter referred to as $w(I')$, is non-increasing throughout the process
    (this implies the total weight of $I$ is also non-increasing since $w(u, v)$ is also non-increasing).
\end{enumerate}

As a first step in proving \cref{clm:dynamic-distributed-bmatching}, we show the following claim is true. 
It roughly states that if in many steps, $(u, v)$ has a large weight compared to the total weight of $I$, then $(u, v)$ is likely to be picked by the process. 

\begin{claim} \label{clm:dynamic-uv-is-picked}
    For an edge $(u, v) \in M^*_1$ with $u \in V(M)$ and $v \in \overline{V(M)}$,
    considering the prefix of steps where $r(u) \geq \epsilon k$
    and $r(v) \geq \epsilon \ceil{kb}$,
    if there are more than $s = 2 \frac{\log(1 / \epsilon)}{\epsilon^6}$ steps where $w(I')$ is less than $W = \floor{\epsilon^{-4}} \cdot k \cdot \ceil{kb}$,
    then $(u, v)$ is picked with probability at least $(1 - \epsilon)$.
\end{claim}
\begin{proof}
    Notice that since $r(u) \geq \epsilon k$ and $r(v) \geq \epsilon \ceil{kb}$,
    it holds that $w(u, v) \geq \epsilon^2 \cdot k \cdot \ceil{kb}$.
    Also, $w(I')$ is at most $W$ in $s$ of the steps.
    Therefore, the probability that $(u, v)$ is \emph{not} picked in any of these $s$ steps is at most:
    \[
    \paren{1 - \frac{\epsilon^2 \cdot k \cdot \ceil{kb}}{\epsilon^2 \cdot k \cdot \ceil{kb} + W}}^s
    \leq \paren{1 - \frac{\epsilon^{-6}}{2}}^s
    \leq \exp\paren{-\frac{s \epsilon^{-6}}{2}}
    \leq \epsilon. \qedhere
    \]
\end{proof}

To complete the proof of \cref{clm:dynamic-distributed-bmatching},
we restrict our attention to the cases where the probability of $(u, v)$ being picked
is smaller than $(1 - \epsilon)$.
Therefore, due to \cref{clm:dynamic-uv-is-picked}
we can assume that among the steps where $r(u) \geq \epsilon k$ and $r(v) \geq \epsilon \ceil{kb}$,
all but $2\frac{\log(1 / \epsilon)}{\epsilon^6}$ of them have $w(I')$ larger than $\floor{\epsilon^{-4}}\cdot k \cdot \ceil{kb}$.
We call them the \emph{early steps}.
Note that these steps form a prefix of the steps since $w(I')$ is non-increasing throughout the process.
We prove that, as a result, with probability $(1 - \epsilon)$ the maximum number of times an edge of $I'$ is picked in these steps 
is at most $\epsilon^3 \ceil{kb}$. 
Intuitively, since $w(I')$ is very large, we expect each edge to be picked $\epsilon^{4}$ fraction of the time, and no edge has the chance to be picked many times, say an $\epsilon^3$ fraction of the time.

To prove this, we characterize the adversary that maximizes the probability of 
the maximum load being larger than $T = \epsilon^3 \ceil{kb}$
after the early steps.
We use $\tau$ to denote this probability and we call an adversary optimal if it maximizes $\tau$.
We say an adversary is greedy if in the early steps, 
it assigns weight $k \cdot \ceil{kb}$ to the $\floor{\epsilon^{-4}}$ edges that have the highest loads (breaking ties arbitrarily), and assigns zero weight to the others,
i.e.\ $w(I')$ is exactly equal to $W$ 
and it is distributed among the edges with the highest loads.
\cref{clm:dynamic-optimal-adversary} states that the greedy adversary is optimal.
Informally, we are stating that the worst thing that can happen
is that the $w(I')$ is always equal to $W$ in the early steps
and the weight is concentrated on the edges with the highest loads.
Throughout the proof, we assume $(u, v)$ is never picked.

\begin{claim} \label{clm:dynamic-optimal-adversary}
Among the adversaries that satisfy the following conditions:
\begin{enumerate}
    \item In every step, every edge has an integer weight in $[0, k \cdot \ceil{kb}]$;
    \item $w(I')$ is non-increasing throughout the process; and
    \item for all but $2\frac{\log(1 / \epsilon)}{\epsilon^6}$ many of the steps such that 
    $r(u) \geq \epsilon k$ and $r(v) \geq \epsilon \ceil{kb}$,
    we have $w(I')$ larger than $W = \floor{\epsilon^{-4}}\cdot k \cdot \ceil{kb}$;
\end{enumerate}
the greedy adversary is optimal.
\end{claim}

\begin{proof}
Let $\card{I'} = q$.
Let $l_1 \geq l_2 \geq \ldots \geq l_q$ be the loads of the edges so far.
We refer to the multiset of loads as the \emph{load profile}.
Note that for an optimal adversary, $\tau$ depends only on the load profile and it does not matter exactly which edge has which load.
We use $J$ to denote the set of the $\floor{\epsilon^{-4}}$ edges with the highest loads.
An adversary is greedy if it sets $w(J) = W$ and $w(I' \setminus J) = 0$
in every step.
This way, an edge from $J$ is picked uniformly at random in every step 
and no other edge is ever picked.
We refer to the $\floor{\epsilon^{-4}}$ highest loads as the \emph{upper load profile}. 
For the greedy adversary, $\tau$ depends only on the upper load profile.
We say an upper load profile 
$L^{(1)}_1 \geq L^{(1)}_2 \geq \ldots \geq L^{(1)}_{\floor{\epsilon^{-4}}}$ dominates another upper load profile
$L^{(2)}_1 \geq L^{(2)}_2 \geq \ldots \geq L^{(2)}_{\floor{\epsilon^{-4}}}$
if for all $i$, it holds that $L^{(1)}_i \geq L^{(2)}_i$.
When $L^{(1)}$ dominates $L^{(2)}$ and the adversary acts greedily,
starting from a $L^{(1)}$ leads to a higher value of $\tau$
than starting from $L^{(2)}$.

We prove the claim by induction on the remaining number of early steps, $n$.
For $n = 1$, the claim is trivial. 
Since there is only one step remaining, $\tau$ is maximized when
the edges with higher loads have the maximum probability of being picked.
Therefore, given any weight assignment,
if there is an edge $e \notin J$ with positive weight,
we can either transfer some of $e$'s weight to $J$ (when $w(J) < W$),
or delete some of $e$'s weight (when $w(J) = W$),
and $\tau$ would grow.

For $n > 1$, we can assume by induction that whatever happens in this step,
from the next step forward, it is optimal for the adversary to act greedily.
Therefore, for the adversaries we examine in the rest of this proof,
we assume they act greedily after the current step.
Now, take any non-greedy weight assignment $w$.
Given that the weights are set to $w$ in this step,
let $\tau(w)$ be the probability that after the early steps finish,
the maximum load is larger than $T = \epsilon^3\ceil{kb}$.
We alter $w$ slightly to obtain a weight assignment $w'$
such that $\tau(w') \geq \tau(w)$ (where $\tau(w')$ is defined similarly to $\tau(w)$). There are two cases.

First, consider the case where $w(J) < W$.
In this case, there must be a an edge $e \notin J$ with $w(e) > 0$,
since we have $w(I') \geq W$.
Take such an edge $e$ with the lowest load.
Also, there must be an edge $e' \in J$ with $w(e') < k \cdot \ceil{kb}$,
otherwise it would have held $w(J) = W$.
We transfer a unit of weight from $e$ to $e'$.
That is, we define $w'(e) = w(e) - 1$, $w'(e') = w(e) + 1$,
and let $w'$ be equal to $w$ for every other edge.
To show $\tau(w') \geq \tau(w)$, loosely we can say
that except for the instances where $e'$ is picked instead of $e$,
the two weight assignments lead to the same outcome.
Therefore, we only need to show that $e'$ being picked instead of $e$ in this step,
leads to a better chance of the maximum load exceeding $T$ in the rest of the early steps.
This is intuitively true because $e'$ has a greater load than $e$.

Formally, each of these two weight assignments leads to a process of $n$ steps.
We introduce a coupling for them as follows:
Consider an outcome of the first process, starting with weights $w$.
If the first process has picked $e$ in the current step,
then with probability $\frac{1}{w(e)}$ (i.e.\ overall probability $\frac{1}{w(I')}$) we assume the second process picks $e'$ in this step and carries on independently of the first process
(this corresponds to the alteration in the weight assignment).
Otherwise, we let the second process have the exact same outcome as the first process.
It can be easily seen that the second process created here,
has the same outcome distribution as an independent process that starts with weights $w'$.

To show $\tau(w') \geq \tau(w)$,
it suffices to prove that when $e'$ is selected in this step,
then the probability of the maximum load going over $T$ in the next $n-1$ steps
is larger than when $e$ is selected.
Because the adversary acts greedily in the next steps,
we only need to examine the upper load profiles.
Let $L$ be the current upper load profile,
let $L^{(1)}$ be the upper load profile resulting from picking $e$,
and $L^{(2)}$ be the upper load profile resulting from picking $e'$.
We use $\tau(L^{(1)})$ to denote the probability of the maximum load
going over $T$ after the next $n - 1$ steps are carried out with the greedy adversary, when the initial load is $L^{(1)}$. We define $\tau(L^{(2)})$ similarly.
We need to show $\tau(L^{(2)}) \geq  \tau(L^{(1)})$.
Let $\gamma$ be $\floor{\epsilon^{-4}}$-th highest load, 
i.e.\ $l_{\floor{\epsilon^{-4}}} = \gamma$.
We consider two cases.
If $l(e) \leq \gamma - 1$,
then when $e$ is picked the upper load profile does not change,
i.e.\ $L^{(1)} = L$, since there are already $\floor{\epsilon^{-4}}$ edges with load larger than $L$.
Meaning that picking $e$ is as good as picking no edges this round
because it will not change the upper load profile.
Also, $L^{(2)}$ dominates $L$.
Therefore, $L^{(2)}$ dominates $L^{(1)}$,
and as a result
$\tau(L^{(2)}) \geq \tau(L^{(1)})$.

Now, consider the case where $l(e) = \gamma$.
In this case, if $e$ is picked, then $l(e)$ becomes $\gamma + 1$.
As a result, in the upper load profile, an element $\gamma$ is replaced with $\gamma + 1$. That is, we have:
\[
L^{(1)} = L \setminus \{l(e)\} \cup \{l(e) + 1\}
\qquad \textnormal{and} \qquad
L^{(2)} = L \setminus \{l(e')\} \cup \{l(e') + 1\}
\]
In \cref{clm:dynamic-better-loads}, we prove that
$\tau(L^{(2)}) \geq \tau(L^{(1)})$.
To apply \cref{clm:dynamic-better-loads}, note that $l(e') \geq l(e)$.
This completes the proof of $\tau(w') \geq \tau(w)$ for when $w(J) < W$.

The claim follows similarly when $w(J) = W$.
There must be an edge $e \notin J$ such that $w(e) > 0$,
otherwise the weight assignment would indeed be greedy.
Take such an edge $e$ with the lowest load.
We define $w'(e) = w(e) - 1$ and let $w'$ be equal to $w$ on all the other edges.
For the coupling, 
when the first process picks edge $e$,
with probability $\frac{1}{w(e)}$
the second process will randomly pick an edge from $I'$
with probability proportional to $w'$.
Otherwise, we let the outcomes be the same.
We still have to examine two cases where $l(e) \leq \gamma - 1$ and $l(e) = \gamma$.
Note that in the case where another edge $e'$ is selected instead of $e$,
it holds that $l(e') \geq l(e)$.

By a series of the two types of alterations we have discussed,
$w$ can be transformed into the greedy assignment of weights.
As proved above, with each alteration, $\tau$ will not decrease.
Therefore, the greedy adversary is optimal for this step as well.
This completes the step of the induction and concludes the proof.
\end{proof}

\begin{claim} \label{clm:dynamic-better-loads}
Let $L_1, \ldots, L_{\floor{\epsilon^{-4}}}$ be an upper load profile.
Let $i$ and $j$ be indices such that $L_i \geq L_j$
and define
\[
    L^{(1)} = \left\{L_1, \ldots, L_i, \ldots, L_j + 1, \ldots, L_{\floor{\epsilon^{-4}}}\right\},
\]
and
\[
    L^{(2)} = \left\{L_1, \ldots, L_i + 1, \ldots, L_j, \ldots, L_{\floor{\epsilon^{-4}}}\right\}.
\]
Then it holds that $\tau(L^{(2)}) \geq \tau(L^{(1)})$.
Where $\tau(L)$ denotes the probability that the maximum load exceeds $T$
after $n$ steps are carried out with the greedy adversary.
\end{claim}\begin{proof}
    We prove the claim by induction. 
    For $n = 0$, it holds trivially since the maximum element of $L^{(2)}$ 
    is at least as large as the maximum element of $L^{(1)}$.
    For $n > 0$, if $L_i = L_j$,
    then $L^{(1)}$ and $L^{(2)}$ are the same multisets.
    Hence $\tau(L^{(1)}) = \tau(L^{(2)})$.

    If $L_i > L_j$, we consider the two processes starting with $L^{(1)}$
    and $L^{(2)}$, and couple them so that they pick the same index in the first step and carry on independently.
    Let the chosen index be $p$, by which we mean the edge with a load equal to $L_p$
    has been picked.
    This leads to two new upper load profiles $L^{(3)}$ for the first process,
    and $L^{(4)}$ for the second process.
    Also, define another upper load profile $L'$ which is equal to $L$, except for $L'_p$ which is equal to $L_p + 1$.
    Then, $L'$, $L^{(3)}$ and $L^{(4)}$ satisfy the conditions of the claim,
    with the same indices $i$ and $j$.
    Therefore, by induction, we have $\tau(L^{(4)}) \geq \tau(L^{(3)})$.
    Hence, whatever happens in this step,
    the second process has a greater chance of exceeding the threshold.
    This completes the step of the induction and concludes the proof.
\end{proof}

With the help of \cref{clm:dynamic-optimal-adversary}, we can prove \cref{clm:dynamic-distributed-bmatching}.
\begin{proof}[Proof of \cref{clm:dynamic-distributed-bmatching}]
Take an adversary and some edge $(u, v) \in M^*_1$.
If $(u, v)$ is picked with probability at least $(1 - \epsilon)$ then the claim holds.
Therefore, we can assume that $(u, v)$ is picked with probability smaller than $(1 - \epsilon)$.
As a result, due to \cref{clm:dynamic-uv-is-picked}, we can assume that 
the adversary lets the total weight be larger than $W = \floor{\epsilon^{-4}}\cdot k \cdot \ceil{kb}$
for all but $s = 2 \frac{\log(1 / \epsilon)}{\epsilon^6}$ steps where $r(u) \geq \epsilon k$ and $r(v) \geq \epsilon \ceil{kb}$.
That is, the adversary satisfies the conditions of \cref{clm:dynamic-optimal-adversary}.
Now, we bound the probability that the maximum load after these steps is larger than $T = \epsilon^3 \ceil{kb}$.
Since we have already established that the greedy adversary is optimal,
it suffices to bound the probability for the greedy adversary.

By the characterization of \cref{clm:dynamic-optimal-adversary},
a fixed set $\floor{\epsilon^{-4}}$ edges of $I'$ have positive weight in all the early steps.
This reduces the problem of bounding the maximum load to an instance of the balls into bins problem.
The probability that a fixed edge $e \in I'$
is picked in a fixed set of $T$ early steps,
is at most $\paren{\frac{k \cdot \ceil{kb}}{W}}^T = \paren{\frac{1}{\floor{\epsilon^{-4}}}}^T$.
Because the weight of each edge is exactly $k \cdot \ceil{kb}$
and the total weight is at least $W$.
Therefore, by taking the union bound over all the possible edges and sets of early steps (there are at most $k + \ceil{kb} \leq 2\ceil{kb}$ steps),
the probability that any edge is picked in more than $T$ early steps,
is at most:
\begin{align*}
\binom{2\ceil{kb}}{T} \floor{\epsilon^{-4}} \paren{\frac{1}{\floor{\epsilon^{-4}}}}^T
&\leq 
\paren{\frac{2e\ceil{kb}}{T}}^T \floor{\epsilon^{-4}} \paren{\frac{1}{\floor{\epsilon^{-4}}}}^T \\
&\leq 
\paren{\frac{2e\ceil{kb}}{T}}^T 2 \epsilon^{-4} \paren{\frac{2}{\epsilon^{-4}}}^T \\
&= \paren{\frac{4e\ceil{kb}}{T \epsilon^{-4}}}^T 2\epsilon^{-4} \\
&\leq (4e\epsilon)^{\epsilon^3 \ceil{kb}} 2\epsilon^{-4} \\
&\leq \epsilon.
\end{align*}
Where the first inequality holds since $\binom{n}{k} \leq \paren{\frac{en}{k}}^k$, and the last inequality holds for small enough $\epsilon$ since $k \geq \epsilon^{-8}$.

Now, assuming $(u, v)$ was not picked during the process, we construct $F$.
Take the endpoint $z$ of $(u, v)$ that first violates $r(z) \geq \epsilon b(z)$.
Before this inequality is violated, there must be at least $(1 - \epsilon)b(z)$ steps where an edge adjacent to $z$ is picked.
At most $s = \frac{\log(1 / \epsilon)}{\epsilon^6} \leq \epsilon k \leq \epsilon b(z)$ of these steps have total weight smaller than $W$,
meaning that the rest of them are early steps.
Therefore, if we just let $F$ equal to the edges adjacent to $z$
that are picked in the early steps, it holds that $\card{F} \geq (1 - 2\epsilon) b(z)$. 
Also, as proved in the last paragraph,
with probability $(1 - \epsilon)$ no edge appears in $F$ more than $\epsilon^3 \ceil{kb}$ times. This concludes the proof.
\end{proof}

Now, we can define the fractional matching $x$ on $M \cup B$.
Recall that we have fixed a maximum matching $M^*$ of $G$,
and let $M^*_1 = M^* \cap (V(M) \times \overline{V(M)})$.
For every edge $e \in B \setminus M^*_1$,
let $t_e$ be $\min(\epsilon^3 \ceil{kb}, B_e)$,
where $B_e$ is the number of occurrences of $e$ in $B$.
For every edge $(u, v) \in B \cap M^*_1$ with
$u \in V(M)$ and $v \in \overline{V(M)}$, we define:
\[
t_e = \min\paren{k - \sum_{e'\in \delta_B(u) \setminus M^*_1}t_{e'}, \ceil{kb} - \sum_{e'\in \delta_B(v) \setminus M^*_1}t_{e'}}.
\]
As a result, for every edge $e$ in $B \setminus M^*_1$ we have $t_e > \epsilon^3 \ceil{kb}$.
Finally, for $e \in B$, we define
\[
x_e = \frac{t_e}{\ceil{kb}},
\]
and for $e \in M$, we define $x_e = 1 - \frac{1}{b}$. 

\begin{claim} \label{clm:dynamic-bmatching-is-saturated}
    For every edge $(u, v) \in M_1^*$ with $u \in V(M)$ and $v \in \overline{V(M)}$, 
    with probability $(1 - \epsilon)$ it holds that
    $t(u) = k$ or $t(v) = \ceil{kb}$.
\end{claim}
\begin{proof}
    The claim follows from \cref{clm:dynamic-distributed-bmatching}.
    When $(u, v)$ is picked in $B$, the claim holds by the definition of $t_{(u, v)}$.
    When $(u, v)$ is not picked in $B$, the claim follows from the existence
    of multiset $F \subseteq B$ as stated in \cref{clm:dynamic-distributed-bmatching}.
\end{proof}

\begin{claim} \label{clm:dynamic-large-bmatching}
    It holds that $\expect{x(B)} \geq (1 - 4\epsilon) \frac{1}{b+1}\card{M^*_1}$.
\end{claim}
\begin{proof}
    The charging argument that we use is quite similar to that of \cref{clm:two-pass-large-bmatching}. We repeat the proof in its entirety for the sake of completeness.
    Order the edges of $B$ arbitrarily as $e_1, \ldots, e_N$,
    and let $B_i$ be the set of first $i$ edges.
    With respect to $B_i$, we define a potential $\phi_i$
    on every edge $(u,v) \in M_1^*$ with $u \in V(M)$ and $v \in \overline{V(M)}$:
    \[
    \phi_i(u, v) = \max\paren{\frac{\displaystyle\sum_{{e \in \delta_{B_i}(u)}}t_{e}}{k}, 
        \frac{\displaystyle\sum_{{e \in \delta_{B_i}(v)}}t_{e}}{\ceil{kb}}},
    \]
    which is equal to the maximum fraction of the used capacity on its endpoints.
    We also let:
    \[
    \phi_i = \sum_{(u,v) \in M_1^*} \phi_i(u,v).
    \]
    
    For an edge $e_i$, we charge it $c_i = \phi_i - \phi_{i-1}$.
    For each edge $e_i$, it holds that $c_i \leq t_{e_i} \cdot \paren{\frac{1}{k} + \frac{1}{\ceil{kb}}}$.
    Because it is adjacent to at most two edges of $M_1^*$,
    and it can increase the potential on either one by at most 
    $\frac{t_{e_i}}{k}$ and $\frac{t_{e_i}}{\ceil{kb}}$ respectively.
    Therefore, we have:
    \begin{equation} \label{eq:dynamic-eq1}
    \phi_N = \sum_{i = 1}^N c_i \leq \paren{\frac{1}{k} + \frac{1}{\ceil{kb}}} \sum_{i = 1}^N t_{e_i}.
    \end{equation}

    We call an edge $(u, v)$ of $M^*_1$ with $u \in V(M)$ and $v \in \overline{V(M)}$, \emph{saturated} if it satisfies $t(u) \geq (1 - 2\epsilon) k$ or $t(v) \geq (1 - 2\epsilon) \ceil{kb}$,
    i.e.\ if it satisfies $\phi_N(u, v) \geq  (1 - 2\epsilon)$.
    If we let $X$ be the number of saturated edges, then by definition it holds:
    \begin{equation} \label{eq:dynamic-eq2}
    \phi_N \geq  (1 - 2\epsilon) X.
    \end{equation}

    Putting \eqref{eq:dynamic-eq1} and \eqref{eq:dynamic-eq2} together we get:
    \[
    (1-2\epsilon)X \leq \paren{\frac{1}{k} + \frac{1}{\ceil{kb}}} \sum_{i = 1}^N t_{e_i},
    \]
    or equivalently:
    \[
    \sum_{i = 1}^N t_{e_i} \geq (1-2\epsilon)\frac{k \cdot \ceil{kb}}{k + \ceil{kb}} X
    = (1-2\epsilon)\frac{k + kb}{k + \ceil{kb}} \cdot \frac{\ceil{kb}}{b+1} X
    \geq (1 - 3\epsilon) \frac{\ceil{kb}}{b + 1} X.
    \]
    
    Due to \cref{clm:dynamic-bmatching-is-saturated},
    each edge in $M^*_1$ is saturated with probability at least $(1 - \epsilon)$.
    Therefore we have $\expect{X} \geq (1 - \epsilon) \card{M^*_1}$.
    Given the fact that $x(e_i) = \frac{t_{e_i}}{\ceil{kb}}$, it follows:
    \[
    \expect{x(B)} = \frac{1}{\ceil{kb}} \expect{\sum_{i = 1}^N t_{e_i}}
    \geq (1 - 3\epsilon) \frac{1}{b + 1}\expect{X} \geq (1 - 4\epsilon) \frac{1}{b + 1} \card{M^*_1}. \qedhere
    \]
\end{proof}

\begin{claim} \label{clm:dynamic-large-maximal}
    It holds that $x(M) \geq \paren{1 - \frac{1}{b}}\paren{\card{M^*_2} + \frac{1}{2}\card{M^*_1}}.$
\end{claim}
\begin{proof}
    The proof is identical to the proof of \cref{clm:two-pass-large-maximal}
    and thus we omit it.
\end{proof}

\begin{claim} \label{clm:dynamic-fractional-gives-integral}
    $M \cup B$ contains an integral matching of size $(1 - \epsilon)^2 \sum_e x_e$.
\end{claim}
\begin{proof}
    The proof is very similar to that of \cref{clm:two-pass-fractional-gives-integral}. The claim follows from \cref{prp:general-fractional-matching}
    since the value of the fractional matching is at most
    $1 - \frac{1}{b}$ on $M$,
    $\frac{1}{b}$ on $B \cap M_1^*$,
    and $\epsilon^3$ on $B \setminus M^*_1$.
\end{proof}

\begin{proof}[Proof of \cref{clm:dynamic-approx}]
 First, we use \cref{clm:dynamic-large-maximal,clm:dynamic-large-bmatching} to show $\expect{\sum_e x_e} \geq (1 - \epsilon)(2 - \sqrt{2})\mu(G)$.
    It holds that:
    \begin{align*}
        \expect{\sum_e x_e} &= x(M) + \expect{x(B)} \\
        &\geq \paren{1 - \frac{1}{b}}\paren{\card{M_2^*} + \frac{1}{2}\card{M_1^*}} + (1 - 4\epsilon) \frac{1}{b + 1}\card{M_1^*} \tag{\cref{clm:dynamic-large-maximal,clm:dynamic-large-bmatching}} \\
        &\geq 
        (1-4\epsilon)\bracket{\paren{1 - \frac{1}{b}}\card{M_2^*} 
        + \paren{\frac{1}{2} -\frac{1}{2b} + \frac{1}{b+1}} \card{M_1^*}}.
    \end{align*}
    Since $b = 1 + \sqrt{2}$,
    we have $1 - \frac{1}{b} = \frac{1}{2} - \frac{1}{2b} + \frac{1}{b + 1} = 2 - \sqrt{2}$.
    Therefore,
    \begin{equation}
    \expect{\sum_e x_e} \geq (1 - 4\epsilon)(2 - \sqrt{2})(\card{M^*_1} + \card{M^*_2})
    = (1 - 4\epsilon)(2 - \sqrt{2}) \mu(G). \label{eq:dynamic-eq3}
    \end{equation}
    
    To complete the proof, we note that by \cref{clm:dynamic-fractional-gives-integral},
    $M \cup B$ contains a matching of size $(1-\epsilon)^2(1-4\epsilon)(2 - \sqrt{2})\mu(G) \geq
    (1 - 6\epsilon) \cdot .585 \cdot \mu(G)$.
    Replacing $\epsilon$ by $\frac{\epsilon}{6}$ concludes the proof.
\end{proof}

\subsubsection*{Lifting the Assumption That the Adversary is Oblivious and $M$ is Maximal}
Thus far, we have taken $M$ to be a maximal matching that we maintain.
For this to be possible, we have relied on the assumption that the adversary is oblivious. Notably, this is the only place where we use this assumption.
It is an open problem to maintain a maximal matching against an adaptive adversary within a $\poly (\log n)$ update time.
Therefore, to resolve this issue, we lift the assumption by taking $M$ to be an \emph{almost} maximal matching instead, meaning $M$ is a maximal matching if we ignore $\epsilon \cdot \mu(G)$ vertices. 
There are existing algorithms that can maintain an almost maximal matching against an adaptive adversary in $\poly(\log n)$ worst-case update time \cite{BKSW2023, wajc2020}.

\begin{proposition}\label{lem:dynamic-almost-maximal-algorithm}
There exists a data structure that maintains an almost maximal matching in a fully dynamic graph with $\poly(\log n)$ worst-case update time against an adaptive adversary.
\end{proposition}

It remains to show that the approximation ratio does not suffer,
i.e.\ \cref{clm:dynamic-approx} holds even when $M$ is an almost maximal matching.
Note that we have only relied on the fact that $M$ is maximal when we use 
$\card{M^*_1} + \card{M^*_2} = \mu(G)$ in \eqref{eq:dynamic-eq3}.
When $M$ is an almost maximal matching,
we have $\card{M^*_1} + \card{M^*_2} \geq (1 - \epsilon) \mu(G)$ instead.
Hence inequality \eqref{eq:dynamic-eq3} still holds with an extra factor of $(1 - \epsilon)$.

\subsubsection*{Lifting The Assumption That $B$ is Maximal}

We have analyzed the algorithm so far assuming that it has access to $B$, a maximal matching $H$. 
However, as stated in \cref{prop:sublinear-maximal},
each edge of $B$ is \enquote{missed} with probability $\epsilon$.
As a result, $\expect{\mu(G[M \cup B])}$ is going to suffer a factor of $(1 - \epsilon)$. Because if we fix any matching in $M \cup B$, 
then the algorithm will successfully find each of its edges with probability $(1 - \epsilon)$.
Assuming $B$ is a maximal matching we have proved $\expect{\mu(G[M \cup B])} \geq (1 - O(\epsilon))(2 - \sqrt{2}) ( \card{M^*_1}  + \card{M^*_2})$.
Therefore, when the edges are missed with probability $\epsilon$,
inequality \cref{eq:dynamic-eq3} still holds with an extra factor of $(1 - \epsilon)$. Therefore, \cref{clm:dynamic-approx} is true.

\subsection{Implementation and Update Time Analysis} \label{sec:dynamic-implementation}

In this subsection, we provide the implementation details and the runtime analysis of \Cref{alg:dynamic}. To facilitate the analysis, we break it down into several smaller parts.

\subsubsection*{Oracle Access to $B$}

To utilize \Cref{prop:local-algorithm} effectively, we require a fast method to obtain all vertices within close proximity to a random vertex $v$. Let $\Delta(G[M \cup B])$ be the maximum degree of the graph that only includes the edges of $M \cup B$. Thus, $\Delta(G[M \cup B]) = O(k)$. Let $d = O(\epsilon^{-3} \log k)$ be the number of rounds needed in \Cref{prop:local-algorithm} for graph $G[M \cup B]$. The following lemma outlines a formal approach to obtaining these vertices efficiently.

\begin{lemma}\label{lem:obtaining-neighborset}
Let $v_1, v_2, \ldots, v_r$ be a set of random vertices in $G$ such that $r = O(\epsilon^{-2}\log n)$. There exists an algorithm that runs in $\widetilde{O}(rk^dn)$ time and for each vertex $v_i$, returns all vertices within distance $d$ of $v_i$ in graph $G[M \cup B]$ with high probability.
\end{lemma}
\begin{proof}
    Let $R = \{v_1, v_2, \ldots, v_r\}$ and $v \in R$. We initiate the process from vertex $v$ and execute a breadth-first search (BFS). Note that for vertex $v$ we do not have the list of its incident edges in $B$ and we only maintain the maximal matching explicitly. To retrieve the list of incident edges of vertex $v$ in $B$, we invoke the oracle defined in \Cref{prop:sublinear-maximal} for each instance of $v$ in $\tH$. This enables us to obtain the adjacency list of vertex $v$ in the graph $G[M \cup B]$. This process is repeated as we execute the BFS, utilizing the oracle as described above whenever we require the adjacency list of a vertex. The process concludes when we have reached all vertices at a distance of $d$. We repeat this process for all vertices of $R$.

    Let $\pi$ be the permutation that we use in the algorithm over edges of $\widetilde H$. Let $l(v, \pi)$ be the number of vertices in a distance of at most $d$ from $v$ in $G[M \cup B]$ and $u^v_1 = v, u^v_2, \ldots, u^v_{l(v, \pi)}$ be all vertices that are visited by BFS if we start from vertex $v$. Thus, $l(v, \pi) = O(k^d)$. Let $T(u, \pi)$ be the time needed by the oracle in \Cref{prop:sublinear-maximal} for vertex $u$ and permutation $\pi$. By \Cref{prop:sublinear-maximal}, we have $\E_{u, \pi}[T(u, \pi)] = \widetilde{O}(n)$. While the expected total running time would be $\widetilde{O}(l(v, \pi)\cdot n)$ if $u^v_1, u^v_2, \ldots, u^v_{l(v, \pi)}$ were chosen uniformly at random (since $l(v, \pi)$ is constant), it is important to note that these vertices do not follow a random selection. Instead, they are the vertices visited by the BFS algorithm starting from a random vertex $v$. Let $S(v, \pi) = \sum_{i=1}^{l(v, \pi)} T(u^v_i, \pi)$. We prove that $\E_{R, \pi}[\sum_{v \in R} S(v, \pi)] = \widetilde{O}(rk^dn)$. Let $\widetilde m$ be the number of edges in $\widetilde H$. We have
    \begin{align*}
        \E_{R, \pi}\left[\sum_{v \in R} S(v, \pi)\right] & = \sum_\pi \sum_R \sum_{v\in R} \sum_{i=1}^{l(v,\pi)} \frac{\E[T(u_i^{v}, \pi)]}{\widetilde m! \cdot {n \choose r}}\\
        &\leq \sum_{\pi}\sum_{v} l(v, \pi) \cdot {n-1 \choose r-1} \cdot \frac{\E[T(v, \pi)]}{\widetilde m! \cdot {n \choose r}} \\
        & \leq O(k^d) \cdot \sum_{\pi} \sum_{v} \frac{{n - 1 \choose r - 1} \cdot \E[T(v, \pi)]}{\widetilde m! \cdot {n \choose r}}\\
        & \leq O(k^d) \cdot \sum_{\pi} \sum_{v} \frac{r \cdot \E[T(v, \pi)]}{\widetilde m! \cdot n}\\
        & = O(rk^d) \cdot \E_{v, \pi}[T(v, \pi)]\\
        & = \widetilde{O}(r k^d  n),
    \end{align*}
    where the last inequality follows by $\E_{v, \pi}[T(v, \pi)] = \widetilde{O}(n)$, which implies that for a random set $R$, BFS takes $\widetilde{O}(rk^dn)$ time in expectation.

In order to achieve a high probability bound on the time complexity, we sample $\Theta(\log n)$ sets of $r$ random vertices $R$, along with a permutation $\pi$. For each of these $\Theta(\log n)$ samples, we execute the described BFS instance. The termination condition is met when, in one of the instances, the BFS halts for all vertices in $R$. Using Markov's inequality, we can deduce that each individual instance terminates within $\widetilde{O}(rk^dn)$ time with a constant probability. Consequently, at least one of these instances terminates within $\widetilde{O}(rk^dn)$ time with high probability. This completes the proof.
\end{proof}

\subsubsection*{Computing the Maximum Matching in $G[M \cup B]$}
Once we have all the vertices within a close distance of vertex $v$, we can use \Cref{prop:local-algorithm} to estimate $\mu(G[M \cup B])$. We design an algorithm that can answer to the query of whether a vertex is matched in $(1-\epsilon)$-approximate matching of $G[M \cup B]$.

\begin{claim}\label{clm:query-lca-expected}
Let $v$ be a random vertex in the graph $G$, and let $D_v$ represent all the vertices in $G[M \cup B]$ that are within a given distance $d$ of $v$. Suppose that subgraph $G[D_v]$ is given. There exists an algorithm that determines if $v$ is matched in $(1-\epsilon)$-approximate matching of $G[M \cup B]$, $\mathcal{L}$, that works in $O(k^d)$ time per query such that if we let $\mathcal{L}(v)$ to be the indicator that shows matching status of $v$, and $\widetilde \ell = \frac{1}{2}\sum_{v \in V} \mathcal{L}(v)$, then we have  $ (1-\epsilon)\cdot \mu(G[M \cup B]) \leq \E[\widetilde \ell] \leq \mu(G[M \cup B])$.
\end{claim}
\begin{proof}
    First, since $\Delta(G[M \cup B]) = O(k)$, $d = O(\epsilon^{-3} \log k)$, we have $|D_v| \leq O(k^d)$. We use the algorithm of \Cref{prop:local-algorithm}. The running time of the algorithm is linear with respect to $|D_v|$. Thus, for a vertex $v$, the running time is $O(k^d)$. Furthermore, since $\widetilde{\ell}$ represents the size of the matching generated by the algorithm described in \Cref{prop:local-algorithm}, it holds $(1-\epsilon) \cdot \mu(G[M \cup B]) \leq \E[\widetilde{\ell}] \leq \mu(G[M \cup B])$.
\end{proof}

\begin{lemma}\label{lem:running-time-final}
    Computing $\widetilde \mu$ in \Cref{alg:dynamic} takes $\widetilde{O}(rk^dn)$ time with high probability.
\end{lemma}
\begin{proof}
    Let $v_1, v_2, \ldots, v_r$ be the sampled vertices in \Cref{ln:number-of-samples} of \Cref{alg:dynamic}. By \Cref{lem:obtaining-neighborset}, the total time to obtain vertices within distance $d$ of $v_i$ for all $i$ in $G[M \cup B]$ takes $\widetilde{O}(rk^dn)$ time with high probability. Let $D_{v_i}$ be all vertices in $G[M \cup B]$ that are within distance $d$ of $v_i$. By \Cref{clm:query-lca-expected}, the time needed to determine if $v_i$ is matched is $O(k^d)$, condition on the fact that $D_{v_i}$ is given which finishes the proof.
\end{proof}

\begin{lemma}\label{lem:approx-final}
    Let $\widetilde \mu$ be the estimate in \Cref{ln:mu-estimate} of \Cref{alg:dynamic}. Then, with high probability, $(2-\sqrt{2} - \epsilon) \cdot \mu(G) - \epsilon n \leq \E[\widetilde \mu] \leq \mu(G)$.
\end{lemma}
\begin{proof}
    Let $\mathcal{L}$ be the algorithm in \Cref{clm:query-lca-expected}, and $v_1, v_2, \ldots, v_r$ be the sampled vertices in \Cref{ln:number-of-samples} of \Cref{alg:dynamic}. Moreover, let $X_i$ be the indicator if $v_i$ is matched by $\mathcal{L}$, i.e.\ $X_i = \mathcal{L}(v_i)$. Let $\widetilde \ell = \frac{1}{2}\sum_{v \in V} \mathcal{L}(v)$. For any realization of $B$ which depends on the permutation $\pi$ over edges of $\widetilde H$, by \Cref{clm:query-lca-expected}, we have $(1-\epsilon)\cdot \mu(G[M \cup B]) \leq \E[\widetilde \ell] \leq \mu(G[M \cup B])$. Hence, 
    \begin{align}\label{eq:ell-mu-B-M}
        (1-\epsilon)\cdot \E[\mu(G[M \cup B])] \leq \E[\widetilde \ell] \leq \E[\mu(G[M \cup B])].
    \end{align}
    Note that $\E[X_i] = 2\E[\widetilde \ell] / n$, given that the number of matched vertices is twice the number of matching edges.  Define $X = \sum_{i=1}^r X_i$. Thus,
    \begin{align}\label{eq:X-to-ell}
        \E[X] = \frac{2r \cdot \E[\widetilde \ell]}{n}.
    \end{align}
    Using Chernoff bound on $X$,
    \begin{align*}
        \Pr[|X - \E[X]| \geq \sqrt{12 \E[X] \log n}] \leq 2 \exp\left(- \frac{12 \E[X] \log n}{3 \E[X]} \right) = \frac{2}{n^4}
    \end{align*}
    Since $\widetilde \mu = (1-\epsilon) \cdot \frac{nX}{2r}$, with probability of $1 - 2/n^4$ we get
    \begin{align*}
            \widetilde \mu & \in \frac{n(\E[X] \pm \sqrt{12 \E[X] \log n})}{2r} - \frac{\epsilon n}{2}\\
            & \in \left( \frac{n\E[X]}{2r} \pm \frac{\sqrt{12n^2 \E[X] \log n}}{2r} \right) - \frac{\epsilon n}{2} \\
            & \in \left(\E[\widetilde\ell] \pm \sqrt{\frac{6n\E[\widetilde\ell] \log n}{r}} \right) - \frac{\epsilon n}{2} & (\text{By \Cref{eq:X-to-ell}})\\
            & \in \left(\E[\widetilde\ell] \pm \sqrt{\frac{\epsilon^2 n\E[\widetilde\ell]}{4}} \right) - \frac{\epsilon n}{2} & (\text{Since }r = 24\epsilon^{-2}\log n)\\
            & \in \E[\widetilde\ell] - \frac{\epsilon n}{2} \pm \frac{\epsilon n}{2} & (\text{Since } \E[\widetilde\ell] \leq n),
    \end{align*}
    thus,
    \begin{align*}
       \E[\widetilde\ell] - \epsilon n \leq \widetilde \mu \leq \E[\widetilde\ell].
    \end{align*}
    Combining with \Cref{eq:ell-mu-B-M},
    \begin{align*}
        (1-\epsilon)\cdot \E[\mu(G[M \cup B])] - \epsilon n\leq \widetilde \mu \leq \E[\mu(G[M \cup B])].
    \end{align*}
    Plugging \Cref{clm:dynamic-approx},
    \begin{align*}
        (1-\epsilon) \cdot (2 - \sqrt{2} - \epsilon) \cdot\mu(G) - \epsilon n\leq \widetilde \mu \leq \mu(G),
    \end{align*}
    and
    \begin{align*}
        (2 - \sqrt{2} - 4\epsilon) \cdot\mu(G) - \epsilon n\leq \widetilde \mu \leq \mu(G),
    \end{align*}
    yeilds the proof using $\epsilon' = \epsilon / 4$ in the algorithm.
\end{proof}

\begin{proof}[Proof of \Cref{thm:final-theorem-dynamic-detail}]
    By \Cref{lem:running-time-final} and \Cref{lem:approx-final}, we have a semi-dynamic algorithm with query time of $\widetilde{O}(rk^dn)$ where $r = O(\epsilon^{-2}\log n)$, $k=O(\epsilon^{-8})$, and $d=O(\epsilon^{-4})$. Also, the worst-case update time of the algorithm is $\poly(\log n)$, which works against an adaptive adversary, and it returns an estimate $\widetilde \mu$ of maximum matching of $G$ such that $(2-\sqrt{2} - \epsilon) \cdot \mu(G) - \epsilon n \leq \E[\widetilde \mu] \leq \mu(G)$. Therefore, the reduction in \Cref{prop:semi-dynamic-to-fully} yields the proof.
\end{proof}

\bibliographystyle{plainnat}
\bibliography{references}

\begin{thebibliography}{42}
\providecommand{\natexlab}[1]{#1}
\providecommand{\url}[1]{\texttt{#1}}
\expandafter\ifx\csname urlstyle\endcsname\relax
  \providecommand{\doi}[1]{doi: #1}\else
  \providecommand{\doi}{doi: \begingroup \urlstyle{rm}\Url}\fi

\bibitem[Abboud and Williams(2014)]{abboud2014popular}
Amir Abboud and Virginia~Vassilevska Williams.
\newblock Popular conjectures imply strong lower bounds for dynamic problems.
\newblock In \emph{2014 IEEE 55th Annual Symposium on Foundations of Computer
  Science}, pages 434--443. IEEE, 2014.

\bibitem[Arar et~al.()Arar, Chechik, Cohen, Stein, and Wajc]{arar2017dynamic}
Moab Arar, Shiri Chechik, Sarel Cohen, Cliff Stein, and David Wajc.
\newblock Dynamic matching: Reducing integral algorithms to
  approximately-maximal fractional algorithms.
\newblock \emph{45th International Colloquium on Automata, Languages, and
  Programming, {ICALP} 2018, July 9-13, 2018, Prague, Czech Republic}.

\bibitem[Assadi et~al.(2019)Assadi, Khanna, and Li]{assadiKhannaYang2019}
Sepehr Assadi, Sanjeev Khanna, and Yang Li.
\newblock The stochastic matching problem with (very) few queries.
\newblock \emph{ACM Trans. Econ. Comput.}, 7\penalty0 (3), sep 2019.

\bibitem[Azarmehr and Behnezhad(2023)]{azarmehr2023robust}
Amir Azarmehr and Soheil Behnezhad.
\newblock Robust communication complexity of matching: Edcs achieves 5/6
  approximation.
\newblock \emph{arXiv preprint arXiv:2305.01070}, 2023.

\bibitem[Baswana et~al.(2018)Baswana, Gupta, and Sen]{BaswanaGS-SJC18}
Surender Baswana, Manoj Gupta, and Sandeep Sen.
\newblock {Fully Dynamic Maximal Matching in $O(\log n)$ Update Time (Corrected
  Version)}.
\newblock \emph{{SIAM} J. Comput.}, 47\penalty0 (3):\penalty0 617--650, 2018.

\bibitem[Behnezhad et~al.(2019)Behnezhad, Derakhshan, Hajiaghayi, Stein, and
  Sudan]{behnezhadDerakhshan19}
S.~Behnezhad, M.~Derakhshan, M.~Hajiaghayi, C.~Stein, and M.~Sudan.
\newblock Fully dynamic maximal independent set with polylogarithmic update
  time.
\newblock In \emph{2019 IEEE 60th Annual Symposium on Foundations of Computer
  Science (FOCS)}, pages 382--405, Los Alamitos, CA, USA, nov 2019. IEEE
  Computer Society.

\bibitem[Behnezhad(2022)]{Behnezhad21}
Soheil Behnezhad.
\newblock Time-optimal sublinear algorithms for matching and vertex cover.
\newblock In \emph{2021 IEEE 62nd Annual Symposium on Foundations of Computer
  Science (FOCS)}, pages 873--884, 2022.

\bibitem[Behnezhad(2023)]{Behnezhad23}
Soheil Behnezhad.
\newblock Dynamic algorithms for maximum matching size.
\newblock In Nikhil Bansal and Viswanath Nagarajan, editors, \emph{Proceedings
  of the 2023 {ACM-SIAM} Symposium on Discrete Algorithms, {SODA} 2023,
  Florence, Italy, January 22-25, 2023}, pages 129--162. {SIAM}, 2023.

\bibitem[Behnezhad and Khanna(2022)]{behnezhad2022new}
Soheil Behnezhad and Sanjeev Khanna.
\newblock New trade-offs for fully dynamic matching via hierarchical edcs.
\newblock In \emph{Proceedings of the 2022 Annual ACM-SIAM Symposium on
  Discrete Algorithms (SODA)}, pages 3529--3566. SIAM, 2022.

\bibitem[Behnezhad et~al.(2020)Behnezhad, {\L}acki, and
  Mirrokni]{behnezhad2020fully}
Soheil Behnezhad, Jakub {\L}acki, and Vahab Mirrokni.
\newblock Fully dynamic matching: Beating 2-approximation in
  ${\Delta}^{\epsilon}$ update time.
\newblock In \emph{Proceedings of the Fourteenth Annual ACM-SIAM Symposium on
  Discrete Algorithms}, pages 2492--2508. SIAM, 2020.

\bibitem[Bernstein and Stein(2015)]{bernstein2015fully}
Aaron Bernstein and Cliff Stein.
\newblock Fully dynamic matching in bipartite graphs.
\newblock In \emph{Automata, Languages, and Programming: 42nd International
  Colloquium, ICALP 2015, Kyoto, Japan, July 6-10, 2015, Proceedings, Part I
  42}, pages 167--179. Springer, 2015.

\bibitem[Bernstein and Stein(2016)]{bernstein2016faster}
Aaron Bernstein and Cliff Stein.
\newblock Faster fully dynamic matchings with small approximation ratios.
\newblock In \emph{Proceedings of the twenty-seventh annual ACM-SIAM symposium
  on Discrete algorithms}, pages 692--711. SIAM, 2016.

\bibitem[Bernstein et~al.(2021{\natexlab{a}})Bernstein, Dudeja, and
  Langley]{bernstein2021framework}
Aaron Bernstein, Aditi Dudeja, and Zachary Langley.
\newblock A framework for dynamic matching in weighted graphs.
\newblock In \emph{Proceedings of the 53rd Annual ACM SIGACT Symposium on
  Theory of Computing}, pages 668--681, 2021{\natexlab{a}}.

\bibitem[Bernstein et~al.(2021{\natexlab{b}})Bernstein, Forster, and
  Henzinger]{bernsteinForster21}
Aaron Bernstein, Sebastian Forster, and Monika Henzinger.
\newblock A deamortization approach for dynamic spanner and dynamic maximal
  matching.
\newblock \emph{ACM Trans. Algorithms}, 17\penalty0 (4), oct
  2021{\natexlab{b}}.

\bibitem[Bhattacharya and Kiss(2021)]{bhattacharya2021deterministic}
Sayan Bhattacharya and Peter Kiss.
\newblock Deterministic rounding of dynamic fractional matchings.
\newblock \emph{arXiv preprint arXiv:2105.01615}, 2021.

\bibitem[Bhattacharya et~al.()Bhattacharya, Kiss, Saranurak, and
  Wajc]{BKSW2023}
Sayan Bhattacharya, Peter Kiss, Thatchaphol Saranurak, and David Wajc.
\newblock \emph{Dynamic Matching with Better-than-2 Approximation in
  Polylogarithmic Update Time}, pages 100--128.

\bibitem[Bhattacharya et~al.(2016)Bhattacharya, Henzinger, and
  Nanongkai]{bhattacharya2016new}
Sayan Bhattacharya, Monika Henzinger, and Danupon Nanongkai.
\newblock New deterministic approximation algorithms for fully dynamic
  matching.
\newblock In \emph{Proceedings of the forty-eighth annual ACM symposium on
  Theory of Computing}, pages 398--411, 2016.

\bibitem[Bhattacharya et~al.(2017)Bhattacharya, Henzinger, and
  Nanongkai]{bhattacharya2017fully}
Sayan Bhattacharya, Monika Henzinger, and Danupon Nanongkai.
\newblock {Fully dynamic approximate maximum matching and minimum vertex cover
  in $O(\log^3 n)$ worst case update time}.
\newblock In \emph{Proceedings of the Twenty-Eighth Annual ACM-SIAM Symposium
  on Discrete Algorithms}, pages 470--489. SIAM, 2017.

\bibitem[Bhattacharya et~al.(2018)Bhattacharya, Henzinger, and
  Italiano]{bhattacharya2018deterministic}
Sayan Bhattacharya, Monika Henzinger, and Giuseppe~F Italiano.
\newblock Deterministic fully dynamic data structures for vertex cover and
  matching.
\newblock \emph{SIAM Journal on Computing}, 47\penalty0 (3):\penalty0 859--887,
  2018.

\bibitem[Charikar and Solomon(2018)]{charikar2018fully}
Moses Charikar and Shay Solomon.
\newblock Fully dynamic almost-maximal matching: Breaking the polynomial
  worst-case time barrier.
\newblock In \emph{45th International Colloquium on Automata, Languages, and
  Programming (ICALP 2018)}. Schloss Dagstuhl-Leibniz-Zentrum fuer Informatik,
  2018.

\bibitem[Dahlgaard(2016)]{Dahlgaard16}
S{\o}ren Dahlgaard.
\newblock On the hardness of partially dynamic graph problems and connections
  to diameter.
\newblock In Ioannis Chatzigiannakis, Michael Mitzenmacher, Yuval Rabani, and
  Davide Sangiorgi, editors, \emph{43rd International Colloquium on Automata,
  Languages, and Programming, {ICALP} 2016, July 11-15, 2016, Rome, Italy},
  volume~55 of \emph{LIPIcs}, pages 48:1--48:14. Schloss Dagstuhl -
  Leibniz-Zentrum f{\"{u}}r Informatik, 2016.

\bibitem[Esfandiari et~al.(2016)Esfandiari, Hajiaghayi, and
  Monemizadeh]{EsfandiariHM16}
Hossein Esfandiari, MohammadTaghi Hajiaghayi, and Morteza Monemizadeh.
\newblock Finding large matchings in semi-streaming.
\newblock In Carlotta Domeniconi, Francesco Gullo, Francesco Bonchi, Josep
  Domingo{-}Ferrer, Ricardo Baeza{-}Yates, Zhi{-}Hua Zhou, and Xindong Wu,
  editors, \emph{{IEEE} International Conference on Data Mining Workshops,
  {ICDM} Workshops 2016, December 12-15, 2016, Barcelona, Spain}, pages
  608--614. {IEEE} Computer Society, 2016.

\bibitem[Feldman and Szarf(2022)]{FeldmanS22}
Moran Feldman and Ariel Szarf.
\newblock Maximum matching sans maximal matching: {A} new approach for finding
  maximum matchings in the data stream model.
\newblock In Amit Chakrabarti and Chaitanya Swamy, editors,
  \emph{Approximation, Randomization, and Combinatorial Optimization.
  Algorithms and Techniques, {APPROX/RANDOM} 2022, September 19-21, 2022,
  University of Illinois, Urbana-Champaign, {USA} (Virtual Conference)}, volume
  245 of \emph{LIPIcs}, pages 33:1--33:24. Schloss Dagstuhl - Leibniz-Zentrum
  f{\"{u}}r Informatik, 2022.

\bibitem[Grandoni et~al.(2022)Grandoni, Schwiegelshohn, Solomon, and
  Uzrad]{grandoni2022maintaining}
Fabrizio Grandoni, Chris Schwiegelshohn, Shay Solomon, and Amitai Uzrad.
\newblock Maintaining an edcs in general graphs: Simpler, density-sensitive and
  with worst-case time bounds.
\newblock In \emph{Symposium on Simplicity in Algorithms (SOSA)}, pages 12--23.
  SIAM, 2022.

\bibitem[Gupta and Peng(2013{\natexlab{a}})]{GuptaPeng13}
Manoj Gupta and Richard Peng.
\newblock Fully dynamic $(1+\epsilon)$-approximate matchings.
\newblock In \emph{2013 IEEE 54th Annual Symposium on Foundations of Computer
  Science}, pages 548--557, 2013{\natexlab{a}}.

\bibitem[Gupta and Peng(2013{\natexlab{b}})]{gupta2013fully}
Manoj Gupta and Richard Peng.
\newblock Fully dynamic (1+ e)-approximate matchings.
\newblock In \emph{2013 IEEE 54th Annual Symposium on Foundations of Computer
  Science}, pages 548--557. IEEE, 2013{\natexlab{b}}.

\bibitem[Harris(2020)]{harris2020}
David~G. Harris.
\newblock Distributed local approximation algorithms for maximum matching in
  graphs and hypergraphs.
\newblock \emph{SIAM Journal on Computing}, 49\penalty0 (4):\penalty0 711--746,
  2020.
\newblock \doi{10.1137/19M1279241}.

\bibitem[Henzinger et~al.(2015)Henzinger, Krinninger, Nanongkai, and
  Saranurak]{henzinger2015unifying}
Monika Henzinger, Sebastian Krinninger, Danupon Nanongkai, and Thatchaphol
  Saranurak.
\newblock Unifying and strengthening hardness for dynamic problems via the
  online matrix-vector multiplication conjecture.
\newblock In \emph{Proceedings of the forty-seventh annual ACM symposium on
  Theory of computing}, pages 21--30, 2015.

\bibitem[Huang et~al.(2019)Huang, Peng, Tang, Tao, Wu, and
  Zhang]{HuangEtal-SODA19}
Zhiyi Huang, Binghui Peng, Zhihao~Gavin Tang, Runzhou Tao, Xiaowei Wu, and
  Yuhao Zhang.
\newblock Tight competitive ratios of classic matching algorithms in the fully
  online model.
\newblock In Timothy~M. Chan, editor, \emph{Proceedings of the Thirtieth Annual
  {ACM-SIAM} Symposium on Discrete Algorithms, {SODA} 2019, San Diego,
  California, USA, January 6-9, 2019}, pages 2875--2886. {SIAM}, 2019.

\bibitem[Kale and Tirodkar(2017)]{KaleT17}
Sagar Kale and Sumedh Tirodkar.
\newblock Maximum matching in two, three, and a few more passes over graph
  streams.
\newblock In Klaus Jansen, Jos{\'{e}} D.~P. Rolim, David Williamson, and
  Santosh~S. Vempala, editors, \emph{Approximation, Randomization, and
  Combinatorial Optimization. Algorithms and Techniques, {APPROX/RANDOM} 2017,
  August 16-18, 2017, Berkeley, CA, {USA}}, volume~81 of \emph{LIPIcs}, pages
  15:1--15:21. Schloss Dagstuhl - Leibniz-Zentrum f{\"{u}}r Informatik, 2017.

\bibitem[Kapralov(2021)]{Kapralov-SODA21}
Michael Kapralov.
\newblock Space lower bounds for approximating maximum matching in the edge
  arrival model.
\newblock In D{\'{a}}niel Marx, editor, \emph{Proceedings of the 2021
  {ACM-SIAM} Symposium on Discrete Algorithms, {SODA} 2021, Virtual Conference,
  January 10 - 13, 2021}, pages 1874--1893. {SIAM}, 2021.

\bibitem[Kiss(2022)]{Kiss22}
Peter Kiss.
\newblock Deterministic dynamic matching in worst-case update time.
\newblock In Mark Braverman, editor, \emph{13th Innovations in Theoretical
  Computer Science Conference, {ITCS} 2022, January 31 - February 3, 2022,
  Berkeley, CA, {USA}}, volume 215 of \emph{LIPIcs}, pages 94:1--94:21. Schloss
  Dagstuhl - Leibniz-Zentrum f{\"{u}}r Informatik, 2022.

\bibitem[Konrad(2018)]{Konrad18}
Christian Konrad.
\newblock A simple augmentation method for matchings with applications to
  streaming algorithms.
\newblock In Igor Potapov, Paul~G. Spirakis, and James Worrell, editors,
  \emph{43rd International Symposium on Mathematical Foundations of Computer
  Science, {MFCS} 2018, August 27-31, 2018, Liverpool, {UK}}, volume 117 of
  \emph{LIPIcs}, pages 74:1--74:16. Schloss Dagstuhl - Leibniz-Zentrum
  f{\"{u}}r Informatik, 2018.

\bibitem[Konrad and Naidu(2021)]{KonradN-APPROX21}
Christian Konrad and Kheeran~K. Naidu.
\newblock On two-pass streaming algorithms for maximum bipartite matching.
\newblock In Mary Wootters and Laura Sanit{\`{a}}, editors,
  \emph{Approximation, Randomization, and Combinatorial Optimization.
  Algorithms and Techniques, {APPROX/RANDOM} 2021, August 16-18, 2021,
  University of Washington, Seattle, Washington, {USA} (Virtual Conference)},
  volume 207 of \emph{LIPIcs}, pages 19:1--19:18. Schloss Dagstuhl -
  Leibniz-Zentrum f{\"{u}}r Informatik, 2021.

\bibitem[Konrad et~al.(2012)Konrad, Magniez, and Mathieu]{KMM2012}
Christian Konrad, Fr{\'e}d{\'e}ric Magniez, and Claire Mathieu.
\newblock Maximum matching in semi-streaming with few passes.
\newblock In Anupam Gupta, Klaus Jansen, Jos{\'e} Rolim, and Rocco Servedio,
  editors, \emph{Approximation, Randomization, and Combinatorial Optimization.
  Algorithms and Techniques}, pages 231--242, Berlin, Heidelberg, 2012.
  Springer Berlin Heidelberg.

\bibitem[Linial(1992)]{Linial92}
Nathan Linial.
\newblock Locality in distributed graph algorithms.
\newblock \emph{{SIAM} J. Comput.}, 21\penalty0 (1):\penalty0 193--201, 1992.
\newblock \doi{10.1137/0221015}.
\newblock URL \url{https://doi.org/10.1137/0221015}.

\bibitem[Neiman and Solomon(2015)]{neiman2015simple}
Ofer Neiman and Shay Solomon.
\newblock Simple deterministic algorithms for fully dynamic maximal matching.
\newblock \emph{ACM Transactions on Algorithms (TALG)}, 12\penalty0
  (1):\penalty0 1--15, 2015.

\bibitem[Onak and Rubinfeld(2010)]{onak2010maintaining}
Krzysztof Onak and Ronitt Rubinfeld.
\newblock Maintaining a large matching and a small vertex cover.
\newblock In \emph{Proceedings of the forty-second ACM symposium on Theory of
  computing}, pages 457--464, 2010.

\bibitem[Roghani et~al.(2022)Roghani, Saberi, and Wajc]{Roghani2022beating}
Mohammad Roghani, Amin Saberi, and David Wajc.
\newblock Beating the folklore algorithm for dynamic matching.
\newblock In Mark Braverman, editor, \emph{13th Innovations in Theoretical
  Computer Science Conference, {ITCS} 2022, January 31 - February 3, 2022,
  Berkeley, CA, {USA}}, volume 215 of \emph{LIPIcs}, pages 111:1--111:23.
  Schloss Dagstuhl - Leibniz-Zentrum f{\"{u}}r Informatik, 2022.

\bibitem[Schrijver et~al.(2003)]{schrijver2003combinatorial}
Alexander Schrijver et~al.
\newblock \emph{Combinatorial optimization: polyhedra and efficiency},
  volume~24.
\newblock Springer, 2003.

\bibitem[Solomon(2016)]{Solomon-FOCS16}
Shay Solomon.
\newblock {Fully Dynamic Maximal Matching in Constant Update Time}.
\newblock In \emph{{IEEE} 57th Annual Symposium on Foundations of Computer
  Science, {FOCS} 2016, 9-11 October 2016, Hyatt Regency, New Brunswick, New
  Jersey, {USA}}, pages 325--334. {IEEE} Computer Society, 2016.

\bibitem[Wajc(2020)]{wajc2020}
David Wajc.
\newblock Rounding dynamic matchings against an adaptive adversary.
\newblock In \emph{Proceedings of the 52nd Annual ACM SIGACT Symposium on
  Theory of Computing}, STOC 2020, page 194–207, New York, NY, USA, 2020.
  Association for Computing Machinery.
\newblock ISBN 9781450369794.

\end{thebibliography}
	
\end{document}